%% file: mergegram_full.tex
\documentclass[a4paper,UKenglish,cleveref, autoref]{lipics-v2019}
\usepackage{amsmath,hyperref,lineno,amssymb,amsthm,mathtools,url,verbatim}
\usepackage{pgfplots}
\pgfplotsset{grid style={red}}
\usepackage{filecontents}
\usepackage{microtype}
\usepackage{float}
\usepackage{tikz}
\usetikzlibrary{matrix}
\usepackage{tikz-cd}
\usepackage{algorithmic, algorithm2e}

\title{Isometry invariant shape recognition of projectively perturbed point clouds by the mergegram extending 0D persistence}\titlerunning{The mergegram of a dendrogram and its stability}

\author{Yury Elkin}{Materials Innovation Factory and Computer Science department, University of Liverpool, UK}{yura.elkin@gmail.com}{}{}

\author{Vitaliy Kurlin}{Materials Innovation Factory and Computer Science department, University of Liverpool, UK}{vitaliy.kurlin@gmail.com}{}{}

\authorrunning{Y.Elkin et. al.}

\Copyright{Yury Elkin, Vitaliy Kurlin}

\ccsdesc[100]{Theory of computation~Computational geometry}

\keywords{Shape Recognition; Topological Data Analysis; Machine Learning; Computer Vision}

\category{}

\relatedversion{}
\nolinenumbers
\supplement{}

\funding{The authors were supported by the £3.5M EPSRC grant EP/R018472/1 (2018-2023) 
}

\EventEditors{}
\EventNoEds{0}
\EventLongTitle{}
\EventShortTitle{}
\EventAcronym{}
\EventYear{}
\EventDate{}
\EventLocation{}
\EventLogo{}
\SeriesVolume{}
\ArticleNo{}
\nolinenumbers 

\usepackage{pgfplots}
\pgfplotsset{compat=1.7}
\pgfplotsset{grid style={red}}
\usepackage{microtype}
\usepackage{tikz}
\usetikzlibrary{matrix}
\usetikzlibrary{arrows,positioning,shapes.geometric}
\usepackage{tikz-cd}
\usepackage{algorithmic, algorithm2e,verbatim}

\usepackage{subcaption}

\theoremstyle{definition}
\newtheorem{dfn}{Definition}[section]
\newtheorem{pro}[dfn]{Problem}
\newtheorem{thm}[dfn]{Theorem}
\newtheorem{exa}[dfn]{Example}

\newtheorem{lem}[dfn]{Lemma}

\newcommand{\R}{\mathbb{R}}
\newcommand{\Z}{\mathbb{Z}}
\newcommand{\PS}{\mathbb{P}}
\newcommand{\V}{\mathbb{V}}
\newcommand{\W}{\mathbb{W}}
\newcommand{\I}{\mathbb{I}}
\newcommand{\id}{\mathrm{id}}
\newcommand{\life}{\mathrm{life}}
\newcommand{\MG}{\mathrm{MG}}
\newcommand{\PD}{\mathrm{PD}}
\newcommand{\HD}{\mathrm{HD}}
\newcommand{\BD}{\mathrm{BD}}
\newcommand{\GH}{\mathrm{GH}}
\newcommand{\ID}{\mathrm{ID}}
\newcommand{\SL}{\mathrm{SL}}
\newcommand{\MST}{\mathrm{MST}}
\newcommand{\birth}{\mathrm{birth}}
\newcommand{\death}{\mathrm{death}}
\newcommand{\de}{\delta}
\newcommand{\De}{\Delta}
\newcommand{\ep}{\epsilon}

\newcommand{\bs}{\hfill $\blacksquare$}

\begin{document}
\maketitle

\begin{abstract}
Rigid shapes should be naturally compared up to rigid motion or isometry, which preserves all inter-point distances.
The same rigid shape can be often represented by noisy point clouds of different sizes.
Hence the isometry shape recognition problem requires methods that are independent of a cloud size.
This paper studies stable-under-noise isometry invariants for the recognition problem stated in the harder form when given clouds can be related by affine or projective transformations.
The first contribution is the stability proof for the invariant mergegram, which completely determines a single-linkage dendrogram in general position.
The second contribution is the experimental demonstration that the mergegram outperforms other invariants in recognizing isometry classes of point clouds extracted from perturbed shapes in images.
\end{abstract}


\section{Introduction: motivations, shape recognition problem and overview of results}
\label{sec:intro}

Real-life objects are often represented by unstructured point clouds obtained by laser range scanning or by selecting salient or feature points in images \cite{pauly2002efficient}.
Point clouds are easy to store and can be used as primitives for visualization \cite{zwicker2002pointshop}.
The above advantages strongly motivate the problem of comparing and classifying unstructured point clouds.
\medskip

Rigid objects are naturally studied up to rigid motion or isometry (including reflections), which is any map that preserves inter-point distances.
The recognition of point clouds of the same number of points is practically solved by the histogram of all pairwise distances, which is a complete isometry invariant in general position \cite{boutin2004reconstructing}.
\medskip

Real shapes are often given in a distorted form because of noisy measurements, when points are perturbed, missed or accidentally added. 
One of the first approaches to recognize nearly identical point clouds $A,B$ of different sizes in the same metric space, for example in $\R^m$, is to use the \emph{Hausdorff} distance  \cite{huttenlocher1993comparing} $\HD(A,B)=\min\ep\geq 0$ such that the first cloud $A$ is covered by $\ep$-balls centered at all points of $B$ and vice versa.  
\medskip

However, we also need to take into account infinitely many potential isometries of the ambient space $\R^m$. 
The exact computation of $\inf_{f} \HD(f(A),B)$ minimized over isometries $f$ of $\R^m$ has a high polynomial complexity already for dimension $m=2$ \cite{chew1997geometric}.
An approximate algorithm is cubic in the number of points for $m=3$ \cite{goodrich1999approximate}.
\medskip

This paper extends the 12-page conference version \cite{elkin2020mergegram}, which introduced the new invariant mergegram but didn't prove its continuity under perturbations.
In addition to the proof of continuity, another contribution is 
Theorem~\ref{thm:mergegram-to-dendrogram} showing how to reconstruct a dendrogram of single-linkage clustering from a mergegram in general position. 
\medskip

The practical novelty is the harder recognition problem including perturbations of isometries within wider classes of affine and projective maps motivated by Computer Vision applications.
Indeed different positions of cameras produce projectively equivalent images of the same rigid shape. 
The new experiments in section~\ref{sec:experiments} extensively compared several approaches on 15000 clouds obtained from real images, see examples in Fig.~\ref{fig:myth_images}.  
\medskip

\newcommand{\scale}{0.13}
\begin{figure}[h!]
    \centering
    \includegraphics[scale = \scale]{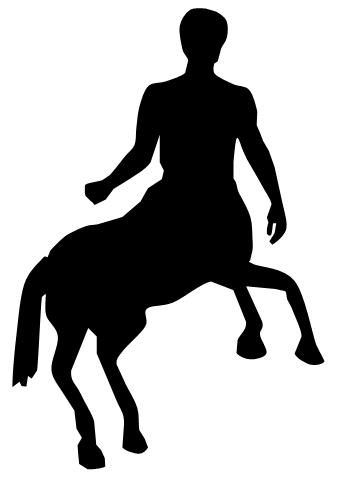}
    \includegraphics[scale = \scale]{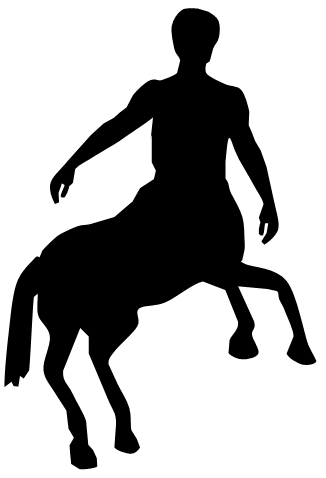}
    \includegraphics[scale = \scale]{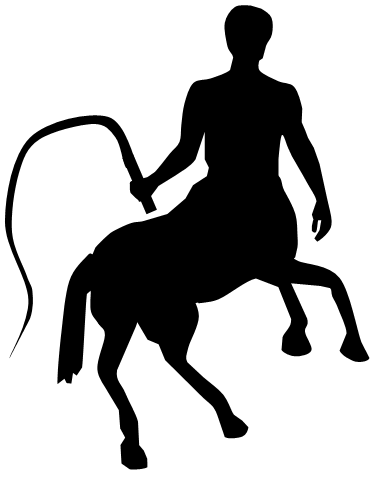}
    \includegraphics[scale = \scale]{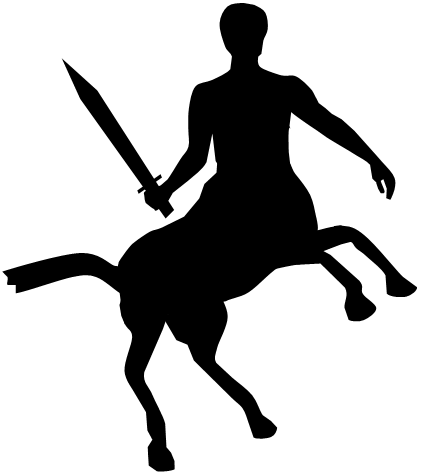}
    \includegraphics[scale = \scale]{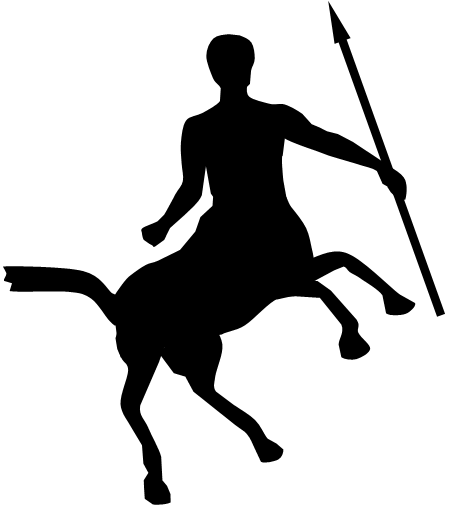}
    \includegraphics[scale = \scale]{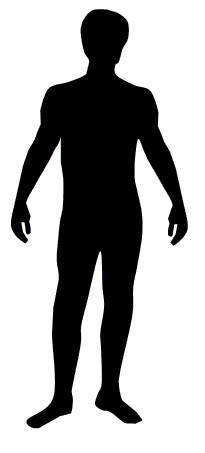}
    \includegraphics[scale = \scale]{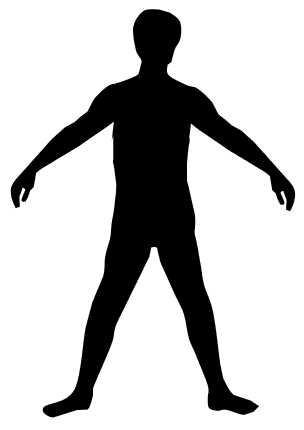}
    \includegraphics[scale = \scale]{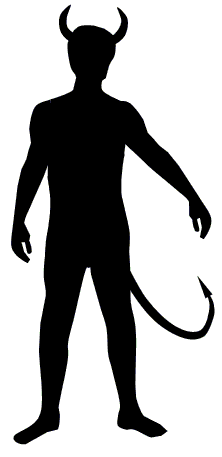}
    \includegraphics[scale = \scale]{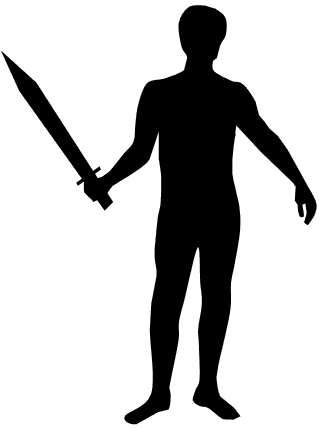}
    \includegraphics[scale = \scale]{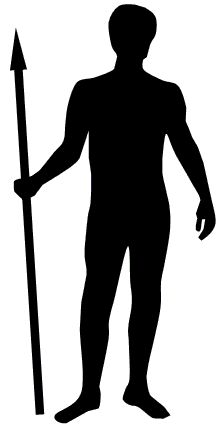}
    \includegraphics[scale = \scale]{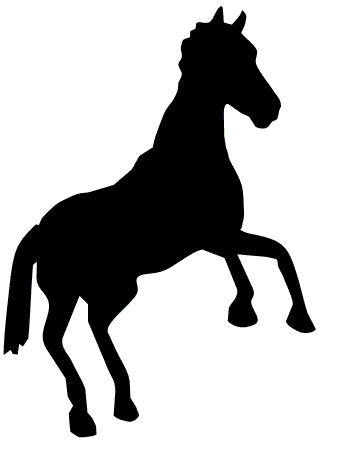}
    \includegraphics[scale = \scale]{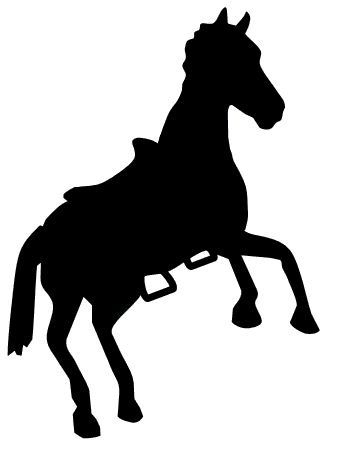}
    \includegraphics[scale = \scale]{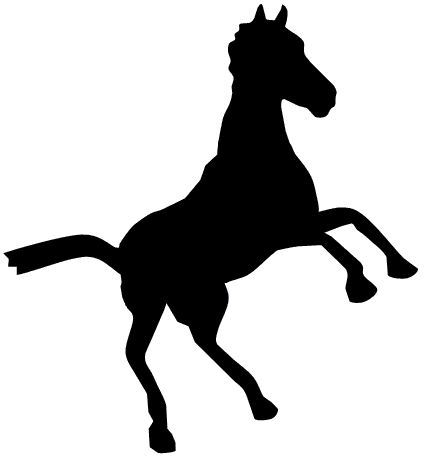}
    \includegraphics[scale = \scale]{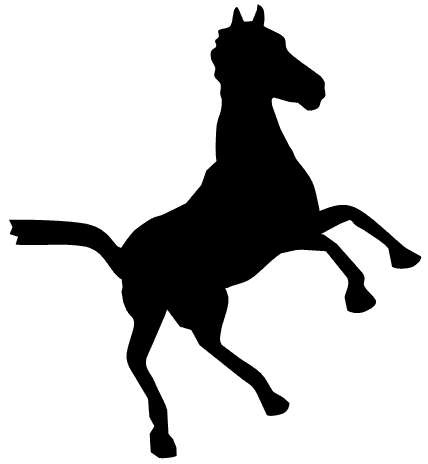}
    \includegraphics[scale = \scale]{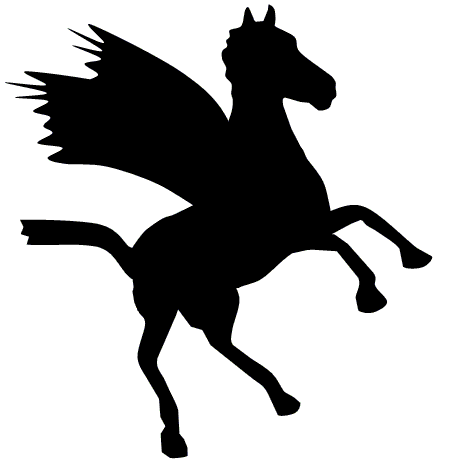}
\caption{Example images from the dataset of mythical creatures, which was introduced in \cite{bronstein2008analysis}.}
    \label{fig:myth_images}
\end{figure}

\begin{pro}[isometry shape recognition under noise]
\label{pro:recognition}
Find an isometry invariant of point clouds in $\R^m$ that is 
(a) independent of a cloud size, (b) provably continuous under perturbations of a cloud, (c) computable in a near linear time in the size of a cloud, and (d) more efficient for recognizing isometry classes of clouds than past invariants.
\bs
\end{pro}

The key contributions are proofs of Theorems~\ref{thm:mergegram-to-dendrogram}, \ref{thm:stability_persistence} and Fig.~\ref{fig:max-layer_affine}, \ref{fig:max-layer_projective}, \ref{fig:image-layer_affine} showing that the mergegram achieves a state-of-the-art recognition on substantially distorted images.

\section{Related work on isometry shape recognition and Topological Data Analysis}
\label{sec:review}

For the isometry classification of clouds consisting of the same number of points, the easiest invariant is the distribution of all pairwise distances, whose completeness (or injectivity) was proved for all point clouds in general (non-singular) position in $\R^m$ \cite{boutin2004reconstructing}.
\medskip

Fixed point clouds $A,B\subset\R^m$ of different sizes can be pairwisely compared by the \emph{Hausdorff} distance \cite{huttenlocher1993comparing}
 $\HD(A,B)=\max\{\sup\limits_{p\in A} d_B(p), \sup\limits_{q\in B} d_A(q)\}$, where $d_B(p)=\inf\limits_{q\in B} d(p,q)$ is the (Euclidean or another) distance from a point $p\in A$ to the cloud $B$.
\medskip

The rigid shape recognition problem for non-fixed clouds $A,B$ is harder because of infinitely many potential isometries that can match $A,B$ exactly or approximately. 
Partial cases of this problem were studied for clouds representing surfaces \cite{elad2001bending} and when two clouds have a given isometric matching of one pair of points \cite{ovsjanikov2010one}.
Shape Google by Ovsjanikov et al. practically extends these ideas to non-rigid shape recognition \cite{ovsjanikov2009shape}.
\medskip

The most general framework for the isometry shape recognition of point cloud data was proposed by M\'emoi and Sapiro \cite{memoli2005theoretical}.
They studied the Gromov-Hausdorff distance $d_{GH}(A,B)=\inf\limits_{f,g,M} \HD(f(A),g(B))$ minimized over all isometric embeddings $f:A\to M$ and $g:B\to M$ of given point clouds into a metric space $M$.
Since the above definitions involve even more minimizations over infinitely many maps and spaces, $\GH$ can be only approximated.
The Farthest Point Sampling (FPS) has a quadratic complexity in the number of points \cite[section~3.6]{memoli2005theoretical} and was successfully tested on small clouds.
\medskip

The proposed invariant mergegram extends the 0-dimensional persistence in the area of Topological Data Analysis (TDA), which grew from the theory of size functions \cite{verri1993use}.
TDA views a point cloud $A\subset\R^m$ not by fixing any distance threshold but across all scales $s$, for example by blurring given points to balls of a variable radius $s$.
The resulting evolution of topological shapes is summarized by a persistence diagram, which is invariant under isometries of $\R^m$.
TDA can be combined with machine learning and statistical tools due to stability under noise, which was first proved by Cohen-Steiner et al. \cite{cohen2007stability} and then extended to a very general form by Chazal et al. \cite{chazal2016structure}.
\medskip

In dimension 0 the \emph{persistence diagram} $\PD(A)$ for distance-based filtrations of a point cloud $A$ consists of the pairs $(0,s)\in\R^2$, where values of $s$ are distance scales at which subsets of $A$ merge by the single-linkage clustering.
These scales equal half-lengths of edges in a \emph{Minimum Spanning Tree} $\MST(A)$.
If distances between all points of $A$ are known, $\MST(A)$ is a connected graph with the vertex set $A$ and a minimal total length.
\medskip

Representing a point cloud $A$ by $\PD(A)$ loses a lot of geometry of $A$, but gains stability under perturbations, which can be expressed in the case of point clouds as $\BD(\PD(A),\PD(B))\leq \HD(A,B)$.
Here the \emph{bottleneck distance} $\BD$ between diagrams is defined as a minimum $\ep\geq 0$ such that all pairs of $\PD(A)$ can be bijectively matched to $\ep$-close points of $\PD(B)$ or to diagonal pairs $(s,s)$, and vice versa.
Here $\ep$-closeness of pairs $(a,c)$ and $(b,d)$ in $\R^2$ is measured in the  distance $L_{\infty}=\max\{|a-b|,|c-d|\}$.
\medskip

The \emph{mergegram} extends $\PD(A)$ to a stronger invariant whose stability under perturbations in the above sense is proved in section~\ref{sec:stability} for the first time. 
The idea of a mergegram is related to the Reeb graph \cite{parsa2013deterministic} or the merge tree \cite{morozov2013interleaving} for the sublevel set filtration of a scalar function. 
The mergegram $\MG$ is defined at a more abstract level for any clustering dendrogram, which can be reconstructed from $\MG$ in general position.
\medskip


Since any persistence diagram and a mergegram are unordered collections of pairs, the experiments in section~\ref{sec:experiments} will use the neural network PersLay \cite{carriere2019perslay} whose output is invariant under permutations of input points by design.
PersLay extends the neural network DeepSet \cite{zaheer2017deep} for unordered sets and introduces new layers to specifically handle persistence diagrams, as well as a new form of representing such permutation-invariant layers.
In other related work deep learning was recently applied to outputs of hierarchical clustering \cite{fu2019learning}, \cite{cirrincione2020gh}, \cite{karim2021deep} and to 0-dimensional persistence \cite{clough2020topological}, \cite{gabrielsson2020topology}. 

\section{Single-linkage clustering and the invariant mergegram of a dendrogram }
\label{sec:mergegram}


\begin{exa}
\label{exa:5-point_line}
Fig.~\ref{fig:5-point_line} illustrates the key concepts before formal Definitions~\ref{dfn:sl_clustering}, \ref{dfn:dendrogram}, \ref{dfn:mergegram} for the point cloud $A = \{0,1,3,7,10\}$ in the real line $\R$.
Imagine that we gradually blur original data points by growing disks of the same radius $s$ around the given points.

\input figures/5-point_line.txt

The disks of the closest points $0,1$ start overlapping at the scale $s=0.5$ when these points merge into one cluster $\{0,1\}$.
This merger is shown by blue arcs joining at the node at $s=0.5$ in the single-linkage dendrogram, see the bottom left picture in Fig.~\ref{fig:5-point_line}.
The persistence diagram $\PD$ in the bottom middle picture of
 Fig.~\ref{fig:5-point_line} represents this merger by the pair $(0,0.5)$ meaning that a singleton cluster of (say) point $1$ was born at the scale $s=0$ and then died later at $s=0.5$ by merging into another cluster of point $0$.
\medskip
 
When clusters $\{0,1,3\}$ and $\{7,10\}$ merge at $s=2$, this event was previously encoded in the persistence diagram by the single pair $(0,2)$ meaning that one cluster inherited from (say) point 7 was born at $s=0$ and died at $s=2$.
The new mergegram in the bottom right picture of Fig.~\ref{fig:5-point_line} represents the above merger by the following two pairs.
The pair $(1,2)$ means that the cluster $\{0,1,3\}$ is merging at the current scale $s=2$ and was previously formed at the smaller scale $s=1$.
The pair $(1.5,2)$ means that another cluster $\{7,10\}$ is merging at the scale $s=2$ and was previously formed at $s=1.5$. 
\medskip

The 0D persistence diagram represents the cluster of the whole cloud $A$  by the pair $(0,+\infty)$, because $A$ was inherited from a singleton cluster starting from $s=0$.
The mergegram represents the same cluster $A$ by the pair $(2,+\infty)$, because $A$ was formed during the last merger of $\{0,1,3\}$ and $\{7,10\}$ at $s=2$ and continues to live as $s\to+\infty$.
\medskip

In the above dendrogram every vertical arc going up from a scale $b$ to $d$ contributes one pair $(b,d)$ to the mergegram.
So both singleton clusters $\{7\}$, $\{10\}$ merging at $s=1.5$ contribute one pair $(0,1.5)$ of multiplicity two shown by two red circles in Fig.~\ref{fig:5-point_line}. 
\end{exa}


\begin{dfn}[single-linkage clustering]
\label{dfn:sl_clustering}
Let $A$ be a finite set in a metric space $X$ with a distance $d:X\times X\to[0,+\infty)$.
Given a distance threshold, which will be called a scale $s$, any points $a,b\in A$ should belong to one \emph{SL cluster} if and only if there is a finite sequence $a=a_1,\dots,a_m=b\in A$ such that any two successive points have a distance at most $s$, so $d(a_i,a_{i+1})\leq s$ for $i=1,\dots,m-1$.
Let $\De_{SL}(A;s)$ denote the collection of SL clusters at the scale $s$.
For $s=0$, any point $a\in A$ forms a singleton cluster $\{a\}$.
Representing each cluster from $\De_{SL}(A;s)$ over all $s\geq 0$ by one point, we get the \emph{single-linkage dendrogram} $\De_{SL}(A)$ visualizing how clusters merge, see the first bottom picture in Fig.~\ref{fig:5-point_line}.
\bs
\end{dfn}

For any $s>0$, all SL clusters $\De_{SL}(A;s)$ can be obtained as connected components of a Minimum Spanning Tree $\MST(A)$ by removing all edges longer than $s$.
\medskip

\begin{dfn}[partition set $\PS(A)$]
\label{dfn:partition}
For any set $A$, a \emph{partition} of $A$ is a finite collection of non-empty disjoint subsets $A_1,\dots,A_k\subset A$ whose union is $A$.
The \emph{single-block} partition of $A$ consists of the set $A$ itself.
The \emph{partition set} $\PS(A)$ consists of all partitions of $A$.
\bs
\end{dfn}

The partition set $\PS(A)$ of the abstract set $A=\{0,1,2\}$ consists of the five partitions
$$(\{0\},\{1\},\{2\}),\quad
(\{0,1\},\{2\}),\quad
(\{0,2\},\{1\}),\quad
(\{1,2\},\{0\}),\quad
(\{0,1,2\}).$$
For example, the collections $(\{0\},\{1\})$ and $(\{0,1\},\{0,2\})$ are not partitions of $A$.
\medskip

Definition~\ref{dfn:dendrogram} below extends a dendrogram from \cite[section~3.1]{carlsson2010characterization} to arbitrary (possibly, infinite) sets $A$.
Since every partition of $A$ is finite by Definition~\ref{dfn:partition}, we don't need to add that an initial partition of $A$ is finite.
Non-singleton sets are now allowed.

\begin{dfn}[dendrogram $\De$ of merge sets]
\label{dfn:dendrogram}
A \emph{dendrogram} $\De$ over any set $A$ is a function $\Delta:[0,+\infty)\to\PS(A)$ of a scale $s\geq 0$ satisfying the following conditions.
\smallskip

\noindent
(\ref{dfn:dendrogram}a)
There exists a scale $r\geq 0$ such that $\De(A;s)$ is the single block partition for $s\geq r$. 
\smallskip

\noindent
(\ref{dfn:dendrogram}b)
If $s\leq t$, then $\De(A;s)$ \emph{refines} $\De(A;t)$, so any set from $\De(t)$ is a subset of some set from $\De(A;t)$.
These inclusions of subsets of $X$ induce the natural map $\De_s^t:\De(s)\to\De(t)$.
\smallskip

\noindent
(\ref{dfn:dendrogram}c)
There are finitely many \emph{merge scales} $s_i$ such that $$s_0 = 0 \text{ and  } s_{i+1} = \text{sup}\{s \mid \text{ the map }  \De_s^t \text{ is identity for } s' \in [s_i,s)\}, i=0,\dots,m-1.$$

\noindent
Since $\De(A;s_{i})\to\De(A;s_{i+1})$ is not an identity map, there is a subset $B\in\De(s_{i+1})$ whose preimage consists of at least two subsets from $\De(s_{i})$.
This subset $B\subset X$ is called a \emph{merge} set and its \emph{birth} scale is $s_i$.
All sets of $\De(A;0)$ are merge sets at the birth scale 0.
The $\life(B)$ is the interval $[s_i,t)$ from its birth scale $s_i$ to its \emph{death} scale $t=\sup\{s \mid \De_{s_i}^s(B)=B\}$.
\bs
\end{dfn}

Dendrograms are usually drawn as trees whose nodes represent all sets from the partitions $\De(A;s_i)$ at merge scales.
Edges of such a tree connect any set $B\in\De(A;s_{i})$ with its preimages under $\De(A;s_{i})\to\De(A;s_{i+1})$.
Fig.~\ref{fig:3-point_dendrogram} shows $\De$ for $A=\{0,1,2\}$.
\medskip

\begin{figure}[h]
\input figures/3-point_dendrogram.txt
\caption{The dendrogram $\De$ on $A=\{0,1,2\}$ and its mergegram  $\MG(\De)$ from Definition~\ref{dfn:mergegram}.}
\label{fig:3-point_dendrogram}
\end{figure}

In Fig.~\ref{fig:3-point_dendrogram} the partition $\De(A;1)$ consists of $\{0,1\}$ and $\{2\}$.
The maps $\De_s^t$ induced by inclusions respect the compositions in the sense that $\De_s^t\circ\De_r^s=\De_r^t$ for any $r\leq s\leq t$.
For example, $\De_0^1(\{0\})=\{0,1\}=\De_0^1(\{1\})$ and $\De_0^1(\{2\})=\{2\}$, so $\De_0^1$ is a well-defined map from the partition $\De(A;0)$ of 3 singleton sets to $\De(A;1)$, but isn't an identity.
\medskip

At the scale $s_0=0$ the merge sets $\{0\},\{1\}$ have $\life=[0,1)$, while the merge set $\{2\}$ has $\life=[0,2)$.
At the scale $s_1=1$ the only merge set $\{0,1\}$ has $\life=[1,2)$.
At the scale $s_2=2$ the only merge set $\{0,1,2\}$ has $\life=[2,+\infty)$.
The notation $\De$ is motivated as the first (Greek) letter in the word dendrogram and by a $\De$-shape of a typical tree.
\medskip

Condition~(\ref{dfn:dendrogram}a) says that 
a partition of a set $X$ is trivial for all large scales.
Condition~(\ref{dfn:dendrogram}b) means that if the scale $s$ is increasing, then sets from a partition $\De(s)$ can only merge but can not split. 
Condition~(\ref{dfn:dendrogram}c) implies that there are only finitely many mergers, when two or more subsets of $X$ merge into a larger merge set.
\medskip

\begin{lem}\cite[Lemma 3.3]{elkin2020mergegram}
\label{lem:sl_clustering}
Given a metric space $(X,d)$ and a finite set $A\subset X$, the single-linkage dendrogram $\De_{SL}(X)$ from Definition~\ref{dfn:sl_clustering} satisfies Definition~\ref{dfn:dendrogram}.
\bs
\end{lem}

A \emph{mergegram} represents life spans of merge sets by pairs 
$(\birth,\death)\in\R^2$.

\begin{dfn}[mergegram $\MG(\De)$]
\label{dfn:mergegram}
The \emph{mergegram} of a dendrogram $\De$ has the pair $(\birth,\death)\in\R^2$ for each merge set $B$ of $\De$ with $\life(B)=[\birth,\death)$.
If any life interval appears $k$ times, the pair (birth,death) has the multiplicity $k$ in $\MG(\De)$.
\bs
\end{dfn}

If our input is a point cloud $A$ in a metric space, then the mergegram $\MG(\De_{SL}(A))$ is an isometry invariant of $A$, because $\De_{SL}(A)$ depends only on inter-point distances.
Though $\De_{SL}(A)$ as any dendrogram is unstable under perturbations of points, the key advantage of $\MG(\De_{SL}(A))$ is its stability, which will be proved in Theorem~\ref{thm:stability_mergegram}.
\medskip

Fig.~\ref{fig:5-point_set} shows the metric space $X=\{a,b,c,d,e\}$ with distances defined by the shortest path metric induced by the specified edge-lengths, see the distance matrix. 

\begin{figure}[H] 
\input figures/5-point_set.txt
\caption{The set $X=\{a,b,c,d,e\}$ has the distance matrix defined by the shortest path metric.}
\label{fig:5-point_set}
\end{figure}

\begin{figure}[H]
\input figures/5-point_mergegram.txt
\caption{\textbf{Left}: the dendrogram $\De$ for the single linkage clustering of the set 5-point set $X=\{a,b,c,d,e\}$ in Fig.~\ref{fig:5-point_set}.
\textbf{Right}: the mergegram $\MG(\De)$ with one pair (0,1) of multiplicity 4.}
\label{fig:5-point_set_mergegram}
\end{figure}

The dendrogram $\De$ in the first picture of Fig.~\ref{fig:5-point_set_mergegram} generates the mergegram as follows:
\begin{itemize}
\item 
each of the singleton sets $\{b\}$, $\{c\}$, $\{p\}$, $\{q\}$ has pair (0,1), so its multiplicity is 4; 
\item 
each of the merge sets $\{b,c\}$ and $\{p,q\}$ has the pair (1,2), so its multiplicity is 2; 
\item 
the singleton set $\{a\}$ has the pair $(0,3)$;
the merge set $\{b,c,p,q\}$ has the pair (2,3);
\item
the full set $\{a,b,c,p,q\}$ continues to leave up to $s=3$, hence has the pair $(3,+\infty)$.
\end{itemize}

\section{Explicit relations between 0-dimensional persistence and mergegram}
\label{sec:relations}

This section recalls the concept of persistence and then shows how any 0D persistence and dendrogram in general position can be reconstructed from a mergegram.

\begin{dfn}[persistence module $\V$]
\label{dfn:persistence_module}
A \emph{persistence module} $\mathbb{V}$ over the real numbers $\mathbb{R}$ is a family of vector spaces $V_t$, $t\in \mathbb{R}$ with linear maps $v^t_s:V_s \rightarrow V_t$, $s\leq t$ such that $v^t_t$ is the identity map on $V_t$ and the composition is respected: $v^t_s \circ v^s_r = v^t_r$ for any $r \leq s \leq t$.
\bs
\end{dfn}

The set of real numbers can be considered as a category  $\mathbb{R}$ in the following sense.
The objects of $\R$ are all real numbers. 
Any two real numbers such that $a\leq b$ define a single morphism $a\to b$.
The composition of morphisms $a\to b$ and $b \to c$ is the morphism $a \leq c$. 
In this language, a persistence module is a functor from $\R$ to the category of vector spaces.
A basic example of a persistence module $\V$ is an interval module.
An interval $J$ between points $p<q$ in $\R$ can be one of the following types: closed $[p,q]$, open $(p,q)$, half-open or half-closed $[p,q)$ and $(p,q]$, all encoded as follows:
$$[p^-,q^+]:=[p,q],\quad
[p^+,q^-]:=(p,q),\quad
[p^+,q^+]:=(p,q],\quad
[p^-,q^-]:=[p,q).$$

The endpoints $p,q$ can also take the infinite values $\pm\infty$, but without superscripts.

\begin{exa}[interval module $\I(J)$]
\label{exa:interval_module}
For any interval $J\subset\R$, the \emph{interval module} $\I(J)$ is the persistence module defined by the following vector spaces $I_s$ and linear maps $i_s^t:I_s\to I_t$
$$I_s=\left\{ \begin{array}{ll} 
\Z_2, & \mbox{ for } s\in J, \\
0, & \mbox{ otherwise }; 
\end{array} \right.\qquad
i_s^t=\left\{ \begin{array}{ll} 
\id, & \mbox{ for } s,t\in J, \\
0, & \mbox{ otherwise }
\end{array} \right.\mbox{ for any }s\leq t.$$
\end{exa}
\medskip

The direct sum $\W=\mathbb{U}\oplus\V$ of persistence modules $\mathbb{U},\V$ is defined  as the persistence module with the vector spaces $W_s=U_s\oplus V_s$ and linear maps $w_s^t=u_s^t\oplus v_s^t$.
\medskip

We illustrate the abstract concepts above by geometric constructions.
Let $f:X\to\R$ be a continuous function on a topological space.
The \emph{sublevel} sets $X_s^f=f^{-1}((-\infty,s])$ form nested subspaces $X_s^f\subset X_t^f$ for any $s\leq t$.
The inclusions of the sublevel sets respect compositions similarly to a dendrogram $\De$ in Definition~\ref{dfn:dendrogram}.
On a metric space $X$ with a metric $d:X\times X\to[0,+\infty)$, a typical example of a function $f:X\to\R$ is the distance $d_A$ to a finite subset $A\subset X$.
For any point $p\in X$, let $d_A(p)$ be the distance from $p$ to a closest point of $A$.
For any $r\geq 0$, the preimage $X_r^{d_A}=d_A^{-1}((-\infty,r])=\{p\in X \mid d_A(p)\leq r\}$ is the union of closed balls with radius $r$ and centers at all points $q\in A$.
For example, $X_0^{d_A}={d_A}^{-1}((-\infty,0])=A$ and $X_{+\infty}^{d_A}={d_A}^{-1}(\R)=X$.
\medskip

If we consider any continuous function $f:X\to\R$, we have the inclusion $X_s^f\subset X_r^f$ for any $s\leq r$.
Hence all sublevel sets $X_s^f$ form a nested sequence of subspaces within $X$.
The above construction of a \emph{filtration} $\{X_s^f\}$ can be considered as a functor from $\R$ to the category of topological spaces.  
Below we discuss the simplest case of dimension 0.

\begin{exa}[persistent homology]
\label{exa:persistent_homology}
For any topological space $X$,  the 0-dimensional \emph{homology} $H_0(X)$ is the vector space (with coefficients $\Z_2$) generated by all connected components of $X$.
Let $\{X_s\}$ be any \emph{filtration} of nested spaces, e.g. sublevel sets $X_s^f$ based on a continuous function $f:X\to\R$.
The inclusions $X_s\subset X_r$ for $s\leq r$ induce the linear maps between homology groups $H_0(X_s)\to H_0(X_r)$ and define the \emph{persistent homology} $\{H_0(X_s)\}$, which satisfies the conditions of a persistence module from Definition~\ref{dfn:persistence_module}.
\bs
\end{exa}
\medskip

If $X$ is a finite set of $m$ points, then $H_0(X)$ is the direct sum $\Z_2^m$ of $m$ copies of $\Z_2$.  
\medskip

The persistence modules that can be decomposed as direct sums of interval modules can be described in a simple combinatorial way by persistence diagrams in $\R^2$.

\begin{dfn}[persistence diagram $\PD(\V)$]
\label{dfn:persistence_diagram}
Let a persistence module $\V$ be decomposed as a direct sum of interval modules : $\V\cong\bigoplus\limits_{l \in L}\I(p^{*}_l,q^{*}_l)$, where $*$ is $+$ or $-$.
The \emph{persistence diagram} $\PD(\V)$ is the multiset 
$\PD(\mathbb{V}) = \{(p_l,q_l) \mid l \in L \} \setminus \{p=q\}\subset\R^2$.
\bs
\end{dfn}
\medskip

The 0-dimensional persistence diagram of a topological space $X$ with a continuous function $f:X\to\R$ is denoted by $\PD\{H_0(X_s^f)\}$.
Lemma~\ref{lem:merge_module_decomposition} will prove that the merge module $M(\De)$ of any dendrogram $\De$ is decomposable into interval modules.
Hence the mergegram $\MG(\De)$ from can be interpreted as the persistence diagram of $M(\De)$. 
\medskip

The following result  describes how the persistence diagram $\PD$ of the distance-based filtration of any point cloud $A$ can be obtained from the mergegram $\MG(\De_{\SL}(S))$.

\begin{thm}\cite[Theorem~5.3]{elkin2020mergegram}
\label{thm:mergegram_to_0D_persistence} 
For a finite set $A$ in a metric space $(X,d)$, let $d_A:X\to\R$ be the distance to $A$.
Let the mergegram $\MG(\De_{SL}(A))$ be a multiset $\{(b_i,d_i)\}_{i=1}^k$, where some pairs can be repeated.
Then the persistence diagram $\PD\{H_0(X_s^{d_A})\}$ is the difference of the multisets $\{(0,d_i)\}_{i=1}^{k}-\{(0,b_i)\}_{i=1}^{k}$ containing each pair $(0,s)$ exactly $\#b-\#d$ times, where $\#b$ is the number of births $b_i=s$ and  $\#d$ is the number of deaths $d_i=s$.
All trivial pairs $(0,0)$ are ignored, alternatively we take $\{(0,d_i)\}_{i=1}^{k}$ only with $d_i>0$.
\bs
\end{thm}

Theorem~\ref{thm:mergegram_to_0D_persistence} is illustrated by Example~\ref{exa:5-point_line}, where $A=\{0,1,3,7,10\}$ has the diagram $\PD(A)=\{(0,0.5),(0,1),(0,1.5),(0,2),(0,+\infty)\}$ obtained from the mergegram
$$\MG(\De_{\SL}(A))=\{(0,0.5),(0,0.5),(0,1),(0,1.5),(0,1.5),(0.5,1),(1,2),(1.5,2),(2,+\infty)\}$$
as follows: one pair $(0,0.5)\in\PD(A)$ comes from two deaths and one birth $s=0.5$ in $\MG(\De_{\SL}(A))$.
Similarly each of the pairs $(0,1),(0,1.5),(0,2)\in\PD(A)$ comes from two deaths and one birth equal to the same scale $s$.
The cloud $B=\{0,4,6,9,10\}\subset\R$ in \cite[Example~1.1]{elkin2020mergegram} has exactly the same $\PD(B)=\PD(A)$ but different $\MG(\De_{\SL}(B))\neq\MG(\De_{\SL}(A))$.
This example together with Theorem~\ref{thm:mergegram_to_0D_persistence} justify that the mergegram is strictly stronger than 0D persistence as an isometry invariant of a point cloud.
\medskip
 

New Reconstruction Theorem~\ref{thm:mergegram-to-dendrogram} below can be contrasted with the weakness of 0D persistence $\PD\{H_0(X_s^{d_A})\}$ consisting of only pairs $(0,s)$ whose finite deaths are half-lengths of edges in a Minimum Spanning Tree $\MST(A)$.
In Example~\ref{exa:5-point_line} these scales $s=0.5,1,1.5,2$ are insufficient to reconstruct the SL dendrogram in Fig.~\ref{fig:5-point_line}. 
Such a unique reconstruction is possible by using the richer invariant mergegram as follows. 

\begin{thm}[from a mergegram to a dendrogram]
\label{thm:mergegram-to-dendrogram}
Let $A$ be a finite point cloud in \emph{general position} in the sense that all merge scales of $A$ in a dendrogram $\De$ from Definition~\ref{dfn:dendrogram} are different.
Then the dendrogram $\De$ can be reconstructed from its mergegram $\MG(\De)$, uniquely up to a permutation of nodes in $\De$ at scale $s=0$.
\bs
\end{thm}
\begin{proof}
Consider all merge scales one by one in the increasing order starting from the smallest.
The general position implies that only two clusters merge at any merge scale. 
For any current scale $s$, the mergegram contains exactly two pairs $(b_1,s)$ and $(b_2,s)$.
\medskip

For a smallest merge scale $s>0$, the births should be $b_1=b_2=0$.
We start drawing a dendrogram $\De$ by merging any two points of $A$ at this smallest scale $s$.
To realize a merger at any larger $s$, we should select two clusters representing $(b_1,s)$ and $(b_2,s)$. 
\medskip

If $b_i=0$ then we take any of the unmerged points of $A$.
If $b_i>0$ then the already constructed dendrogram should contain a unique non-singleton cluster determined by the scale $b_i\in(0,s)$.
Hence at any merge scale $s$ we know how to select two clusters to merge.
The only choice comes from choosing points of $A$ or permuting notes of $\De$.
\end{proof}

Following the above proof of Theorem~\ref{thm:mergegram-to-dendrogram} for the cloud $A=\{0,1,3,7,10\}$ in Example~\ref{exa:5-point_line}, the first two pairs $(0,0.5)\in\MG(\De_{\SL}(A))$ indicate that we should merge two points of $A$ at $s=0.5$.
The scale $s=0.5$ uniquely determines this 2-point cluster. 
\medskip

The next two pairs $(0,1),(0.5,1)$ mean that the above cluster born at $s=0.5$ should merge at $s=1$ with a singleton cluster (any free point of $A$).
The resulting 3-point cluster is uniquely determined by its merge scale $s=1$.
The further two pairs $(0,1.5),(0,1.5)$ say that a new 2-point cluster is formed at $s=1.5$ by the two remaining points of $A$.
\medskip

The final pairs $(1,2),(1.5,2)$ tell us to merge at $s=2$ the two clusters formed earlier at $s=1$ and $s=1.5$.
The resulting dendrogram $\De$ has the expected combinatorial structure as in Fig.~\ref{fig:5-point_line}, though we can draw $\De$ in another way by permuting points of $A$.

\section{Stability of the mergegram for any single-linkage dendrogram}
\label{sec:stability}

This section fully proves the stability of a mergegram, which was stated in \cite[Theorem~7.4]{elkin2020mergegram} without proving key Lemmas~\ref{lem:merge_module_decomposition} and \ref{lem:merge_modules_interleaved}.
For simplicity, we consider vector spaces with coefficients only in $\Z_2=\{0,1\}$, which can be replaced by any field.
\medskip

Definition~\ref{dfn:homo_modules} introduces homomorphisms between persistence modules, which are needed to state the stability of persistence diagrams $\PD\{H_0(X_s^f)\}$ under perturbations of a function $f:X\to\R$.
This result will imply a stability for the mergegram $\MG(\De_{SL}(A))$ for the dendrogram $\De_{SL}(A)$ of the single-linkage clustering of a set $A\subset X$.

\begin{dfn}[a homomorphism of a degree $\de$ between persistence modules]
\label{dfn:homo_modules}
Let $\mathbb{U}$ and $\mathbb{V}$ be persistent modules over $\mathbb{R}$. 
A \emph{homomorphism} $\mathbb{U}\to\V$ of \emph{degree} $\delta\in\R$ is a collection of linear maps $\phi_t:U_t \rightarrow V_{t+\delta}$, $t \in \mathbb{R}$, such that the diagram commutes for all $s \leq t$. 
\begin{figure}[H]
\centering
\begin{tikzpicture}[scale=1.0]
  \matrix (m) [matrix of math nodes,row sep=3em,column sep=4em,minimum width=2em]
  {
     U_s & U_t \\
     V_{s+\delta} & V_{t+\delta} \\};
  \path[-stealth]
    (m-1-1) edge node [left] {$\phi_s$} (m-2-1)
            edge [-] node [above] {$u^t_s$} (m-1-2)
    (m-2-1.east|-m-2-2) edge node [above] {$v^{t+\delta}_{s+\delta}$}
            node [above] {} (m-2-2)
      (m-1-2) edge node [right] {$\phi_t$} (m-2-2);
\end{tikzpicture}
\end{figure}
Let $\text{Hom}^\delta(\mathbb{U},\mathbb{V})$ be all homomorphisms $\mathbb{U}\rightarrow \mathbb{V}$  of degree $\delta$.
Persistence modules $\mathbb{U},\V$ are \emph{isomorphic} if they have inverse homomorphisms $\mathbb{U}\to\V\to\mathbb{U}$ of degree $0$.
\bs
\end{dfn}

For a persistence module $\V$ with maps $v_s^t:V_s\to V_t$, the simplest example of a homomorphism of a degree $\de\geq 0$
 is $1_{\V}^{\de}:\V\to\V$ defined by the maps $v_s^{s+\de}$, $t\in\R$.
So $v_s^t$ defining the structure of $\V$ shift all vector spaces $V_s$ by the difference $\de=t-s$.
\medskip

The concept of interleaved modules below is an algebraic generalization of a geometric perturbation of a set $X$ in terms of (the homology of) its sublevel sets $X_s$.

\begin{dfn}[interleaving distance ID]
\label{dfn:interleaving_distance}
Persistence modules $\mathbb{U}$ and $\mathbb{V}$ are called $\delta$-interleaved if there are homomorphisms $\phi\in \text{Hom}^\delta(\mathbb{U},\mathbb{V})$ and $\psi \in \text{Hom}^\delta(\mathbb{V},\mathbb{U}) $ such that $\phi\circ\psi = 1_{\mathbb{V}}^{2\de} \text{ and } \psi\circ\phi = 1_{\mathbb{U}}^{2\de}$.
The \emph{interleaving distance} between the persistence modules $\mathbb{U}$ and $\mathbb{V}$ is 
$\ID(\mathbb{U},\V)=\inf\{\de\geq 0 \mid \mathbb{U} \text{ and } \mathbb{V} \text{ are } \delta\text{-interleaved} \}.$
\bs
\end{dfn}

If $f,g:X\to\R$ are continuous functions such that $||f-g||_{\infty}\leq\de$ in the $L_{\infty}$-distance, the modules $H_k\{f^{-1}(-\infty,s]\}$, $H_k\{g^{-1}(-\infty,s]\}$ are $\de$-interleaved for any $k$ \cite{cohen2007stability}.
The last conclusion extends to persistence diagrams for the bottleneck distance below.

\begin{dfn}[bottleneck distance BD]
\label{dfn:bottleneck_distance}
Let multisets $C,D$ contain finitely many points $(p,q)\in\R^2$, $p<q$, of finite multiplicity and all diagonal points $(p,p)\in\R^2$ of infinite multiplicity.
For $\de\geq 0$, a $\de$-matching is a bijection $h:C\to D$ such that $|h(a)-a|_{\infty}\leq\de$ in the $L_{\infty}$-distance for any point $a\in C$.
The \emph{bottleneck} distance between persistence modules $\mathbb{U},\V$ is $\BD(\mathbb{U},\mathbb{V}) = \text{inf}\{ \delta \mid \text{ there is a }\delta\text{-matching between } \PD(\mathbb{U}), \PD(\mathbb{V})\}$. 
\bs
\end{dfn}

The original stability of persistence for sequences of sublevel sets was extended as Theorem~\ref{thm:stability_persistence} to $q$-tame persistence modules. 
A persistence module $\V$ is $q$-tame if any non-diagonal square in the persistence diagram $\PD(\V)$ contains only finitely many of points, see \cite[section~2.8]{chazal2016structure}.  
Any finitely decomposable persistence module is $q$-tame.
  
\begin{thm}[stability of persistence modules]\cite[isometry theorem~4.11]{chazal2016structure}
\label{thm:stability_persistence}
 Let $\mathbb{U}$ and $\mathbb{V}$ be q-tame persistence modules. Then $\ID(\mathbb{U},\mathbb{V}) = \BD(\PD(\mathbb{U}),\PD(\mathbb{V}))$,
 where $\ID$ is the interleaving distance, $\BD$ is the bottleneck distance between persistence modules.
\bs
\end{thm}

\begin{dfn}[merge module $M(\De)$]
\label{dfn:merge_module}
For any dendrogam $\De$ on a finite set $X$,
the \emph{merge module} $M(\De)$ consists of the vector spaces $M_s(\De)$, $s\in\R$, and linear maps $m_s^t:M_s(\De)\to M_t(\De)$, $s\leq t$.
For any $s\in\R$ and $A\in\De(s)$, the space $M_s(\De)$ has the generator or a basis vector $[A]\in M_s(\De)$.
For $s<t$ and any set $A\in\De(s)$, 
if the image of $A$ under $\De_s^t$ coincides with $A\subset X$, so $\De_s^t(A)=A$, then $m_s^t([A])=[A]$, else $m_s^t([A])=0$. 
\bs
\end{dfn}

\begin{figure}[h]
\begin{tabular}{lccccccccc}
scale $s_3=+\infty$ & 0 & & & & & 0 \\
map $m_2^{+\infty}$ & $\uparrow$ & & & & & $\uparrow$\\
scale $s_2=2$ & $\Z_2$ & & & 0 & 0 & [\{0,1,2\}]\\
map $m_1^2$ & $\uparrow$ & & & $\uparrow$ & $\uparrow$\\
scale $s_1=1$ & $\Z_2\oplus\Z_2$ & 0 & 0 & [\{2\}] & [\{0,1\}] \\
map $m_0^1$ & $\uparrow$ & $\uparrow$ & $\uparrow$ & $\uparrow$ \\
scale $s_0=0$ & $\Z_2\oplus\Z_2\oplus\Z_2$ & [\{0\}] & [\{1\}] & [\{2\}] &
\end{tabular}
\caption{The merge module $M(\De)$ of the dendrogram $\De$ on the set $X=\{0,1,2\}$ in Fig.~\ref{fig:3-point_dendrogram}.}
\label{fig:3-point_module}
\end{figure}

In a dendrogram $\De$ from Definition~\ref{dfn:dendrogram}, any merge set $A$ of $\De$ has $\life(A)=[b,d)$ from its birth scale $b$ to its death scale $d$.
Lemmas~\ref{lem:merge_module_decomposition} and~\ref{lem:merge_modules_interleaved} are proved for the first time. 

\begin{lem}[merge module decomposition]
\label{lem:merge_module_decomposition}
For any dendrogram $\De$ from Definition~\ref{dfn:dendrogram}, the merge module $M(\De)\cong\bigoplus\limits_{A}\mathbb{I}(\life(A))$ decomposes over all merge sets $A$.
\bs
\end{lem}
\begin{proof}[Proof of Lemma~\ref{lem:merge_module_decomposition}]
The goal is to prove that $M(\triangle) \cong \bigoplus_{A}\mathbb{I}(\text{life}(A))$. 
Recall that the interval module $\mathbb{I}(\text{life}(A))$
 consists of only vector spaces $0$ and $\Z_2$.
For a scale $r$, let $\mathbb{I}_r(\text{life}(A))$ be its vector space, whose generator is denoted by $[\mathbb{I}_r(\text{life}(A))]$.  
Define
$$\psi_r:M_r(\triangle) \rightarrow \bigoplus_{A}\mathbb{I}_r(\text{life}(A)) \text{ such that } [A] \rightarrow [\mathbb{I}_r(\text{life}(A))] \text{ for all }A \in \triangle(r),$$
$$\phi_r:\bigoplus_{A}\mathbb{I}_r(\text{life}(A)) \rightarrow M_r(\Delta) \text{ such that } [\mathbb{I}_r(\text{life}(A))] \rightarrow [A] \text{ for all }\text{life}(A) \text{ containing }r.$$
We will first prove that $\phi_r$ is well-defined. 
If $r \in \text{life}(A)$ then $A \in M_r(\triangle)$.
We know that $M_r(\triangle)$ is generated by elements $A \in \triangle(r)$ for which $r \in \text{life}(A)$. 
Thus the compositions satisfy $\phi_r \circ \psi_r = \text{id}_r$ and $\psi_r \circ \phi_r = \text{id}_r$.
It remains to prove that morphisms correctly behave under the functors $\psi, \phi$. 
The proofs for both cases are essentially the same, thus we will prove it only for $\psi$. 
The goal is to prove that the following diagram commutes:

\begin{figure}[H]
\centering 
\begin{tikzpicture}[scale=1.0]
  \matrix (m) [matrix of math nodes,row sep=3em,column sep=4em,minimum width=2em]
  {
     M_s(\triangle) & M_t(\triangle) \\
     \bigoplus_{A}\I_s(\text{life}(A)) & \bigoplus_{A}\I_t(\text{life}(A)) \\};
  \path[-stealth]
    (m-1-1) edge node [left] {$\psi_s$} (m-2-1)
            edge [->] node [above] {$m^t_s$} (m-1-2)
    (m-2-1.east|-m-2-2) edge node [above] {$i^t_s$}
            node [above] {} (m-2-2)
      (m-1-2) edge node [right] {$\psi_t$} (m-2-2);
\end{tikzpicture}
\end{figure}

Here $i^t_s$ is the direct sum of the corresponding maps of interval modules $\bigoplus_{A}(i^t_s)^A $ . 
Let $[A]$ be an arbitrary generator of $M_r(\triangle)$. 
There are two possibilities how $m^t_s$ can map $[A]$. 
If $t \in \text{life}(A)$, then $m^t_s([A]) = [A] \in M_t(\triangle)$ and by definition $$\phi_t \circ m^t_s([A]) = [\mathbb{I}_t(\text{life}(A))].$$
Since both $s,t \in \text{life}(A)$, we also have that
$$m^t_s \circ \phi_t([A]) = [\mathbb{I}_t(\text{life}(A))] = \phi_t \circ m^t_s([A]).$$
Assume now that $t \notin \text{life}(A)$. 
Then $m^t_s([A]) = 0$ and thus $\phi_t(m^t_s([A])) = 0$. On the other hand $i^t_s\circ\phi_s([A]) = [\mathbb{I}_t(\text{life}(A))] = \mathbb{Z}_2$.
Since $t \notin \text{life}(A)$, we get $i^t_s\circ\phi_s([A]) = 0$.
\end{proof}


\begin{lem}[merge modules interleaved]
\label{lem:merge_modules_interleaved}
If subsets $A,B$ of a metric space $(X,d)$ have $\HD(A,B)=\de$, then the merge modules $M(\De_{SL}(A))$, $M(\De_{SL}(B))$ are $\de$-interleaved.
\bs
\end{lem}
\begin{proof}[Proof of Lemma~\ref{lem:merge_modules_interleaved}]
The equality $\HD(A,B)=\de$ means that $A$ is covered by the union of closed balls that have the radius $\de$ and centers at all points $b\in B$.
This union is the preimage is $d_B^{-1}([0,\de])$, i.e. $A\subset d_B^{-1}([0,\de])$.
Extending the distance values by $s\geq 0$, we get $d_A^{-1}([0,s])\subset d_B^{-1}([0,s+\de])$ and similarly $d_B^{-1}([0,s])\subset d_A^{-1}([0,s+\de])$.
\medskip

Let $U$ be an arbitrary set in $\De_{SL}(A)$.
Define map $\phi_r:M(A; r) \rightarrow M(B; r+\delta)$ 
$$\phi_r([U]) = \left\{ \begin{array}{ll} 
[U], & \mbox{ if } r+\delta\in \text{life}(\De_{SL}(B),U), \\
0, & \mbox{ otherwise }.
\end{array} \right.\qquad $$

Symmetrically for any $V \in \De_{SL}(B)$ we define $\psi_r:M(B; r) \rightarrow M(A; r+\delta) $ 

$$\psi_r([V]) = \left\{ \begin{array}{ll} 
[V], & \mbox{ if } r+\delta\in \text{life}(\De_{SL}(A),V), \\
0, & \mbox{ otherwise }.
\end{array} \right.\qquad $$
In the notation above, $\text{life}(\De_{SL}(B),U)$ is the $\text{life}(U)$ in the dendrogram $\De_{SL}(B)$. 
If $U \notin \De_{SL}(B)(t)$ for all values $t$, then $\text{life}(U) = \emptyset$.
By symmetry it is enough to prove that the following diagrams commute:
\begin{figure}[H]
\centering
\begin{tikzpicture}[scale=1.2]
  \matrix (m) [matrix of math nodes,row sep=3em,column sep=4em,minimum width=2em]
  {
     M_s(\De_{SL}(A)) & M_t(\De_{SL}(A)) \\
     M_{s+\delta}(\De_{SL}(B)) & M_{t+\delta}(\De_{SL}(B)) \\};
  \path[-stealth]
    (m-1-1) edge node [left] {$\phi_s$} (m-2-1)
            edge [->] node [above] {$m^t_s$} (m-1-2)
    (m-2-1.east|-m-2-2) edge node [above] {$m^{t+\delta}_{s+\delta}$}
            node [above] {} (m-2-2)
      (m-1-2) edge node [right] {$\phi_t$} (m-2-2);
\end{tikzpicture}
\medskip

\begin{tikzcd}
 & M_s(\De_{SL}(B)) \arrow["\psi_s"]{dr}{} \\
M_{s-\delta}(\De_{SL}(A)) \arrow["\phi_{s-\delta}"]{ur}{} \arrow["m^{s+\delta}_{s-\delta}"]{rr}{} && M_{s+\delta}(\De_{SL}(A))
\end{tikzcd}
\end{figure}

We note first that if $[a,b) = (\text{life}(\De_{SL}(A),U)$, then $(\text{life}(\De_{SL}(B),U) \subseteq [a-\epsilon, b+\epsilon)$
\medskip

We begin by proving commutativity of the first diagram. Let $U$ be arbitrary element of $\De_{SL}(A)(s)$. If $t \notin \text{life}(\De_{SL}(A),U)$ then $\phi_t \circ m^t_s = 0$. If $s+\delta \notin \text{life}(\De_{SL}(B),U)$ or $t+\delta \notin \text{life}(\De_{SL}(B),U)$ then we are done. 
Since $t \notin \text{life}(\De_{SL}(A),U)$, it follows that $t+\delta \notin \text{life}(\De_{SL}(B),U)$.
And thus with given assumptions the diagram commutes.
\medskip

Assume now that $t+\delta \notin \text{life}(\De_{SL}(A),U) $. Then both $\phi_t(m^t_s(U)) = 0 = m^{t+\delta}_{s+\delta}(\phi_s(U))$. 
In the last case we assume that $t \in \text{life}(\De_{SL}(A),U)$ and $t+\delta \in \text{life}(\De_{SL}(B),U)$. 
In this case obviously $s+\delta \in \text{life}(\De_{SL}(B),U)$ and thus $\phi_t(m^t_s([U])) = [U] = m^{t+\delta}_{s+\delta}(\phi_s([U]))$.
\medskip

For the second diagram assume now that $U \in M(\De_{SL}(A))(s-\delta) $. 
Assume first that $s \notin \text{life}(\De_{SL}(B),U)$, then $s+\delta \notin \text{life}(\De_{SL}(B),U)$ and $m^{s+\delta}_{s-\delta}([U]) = 0 = \psi_s(\phi_{s-\delta}[U])$.
\medskip

Assume then that $s \in \text{life}(\De_{SL}(B),U)$. Now outcome of both maps $\psi_s$ and $m^{s+\delta}_{s-\delta}$ depend on if $s+\delta \in \text{life}(\De_{SL}(A),U)$ and thus $m^{s+\delta}_{s-\delta}([U]) = \psi_s(\phi_{s-\delta}[U])$.
Since all the diagrams commute, the required conclusion follows.
\end{proof}

\begin{thm}[stability of a mergegram]
\label{thm:stability_mergegram}
The mergegrams of any finite point clouds $A,B$ in a metric space $(X,d)$ satisfy $\BD(\MG(\De_{SL}(A)),\MG(\De_{SL}(B))\leq \HD(A,B)$.
Hence any small perturbation of a cloud $A$ in the Hausdorff distance yields a similarly small perturbation in the bottleneck distance for its mergegram $\MG(\De_{SL}(A))$.
\bs
\end{thm}
\begin{proof}
The given clouds $A,B\subset X$ with $\HD(A,B)=\de$ have $\de$-interleaved merge modules by Lemma~\ref{lem:merge_modules_interleaved}, so $\ID(\MG(\De_{SL}(A)),\MG(\De_{SL}(B))\leq\de$.
Since any merge module $M(\De)$ is finitely decomposable, hence is $q$-tame by Lemma~\ref{lem:merge_module_decomposition}. 
The corresponding mergegram $\MG(M(\De))$ satisfies Theorem~\ref{thm:stability_persistence}, so
$\BD(\MG(\De_{SL}(A)),\MG(\De_{SL}(B))\leq\de$.
\end{proof}

\begin{figure}[h!]
\centering
\includegraphics[width=0.49\linewidth]{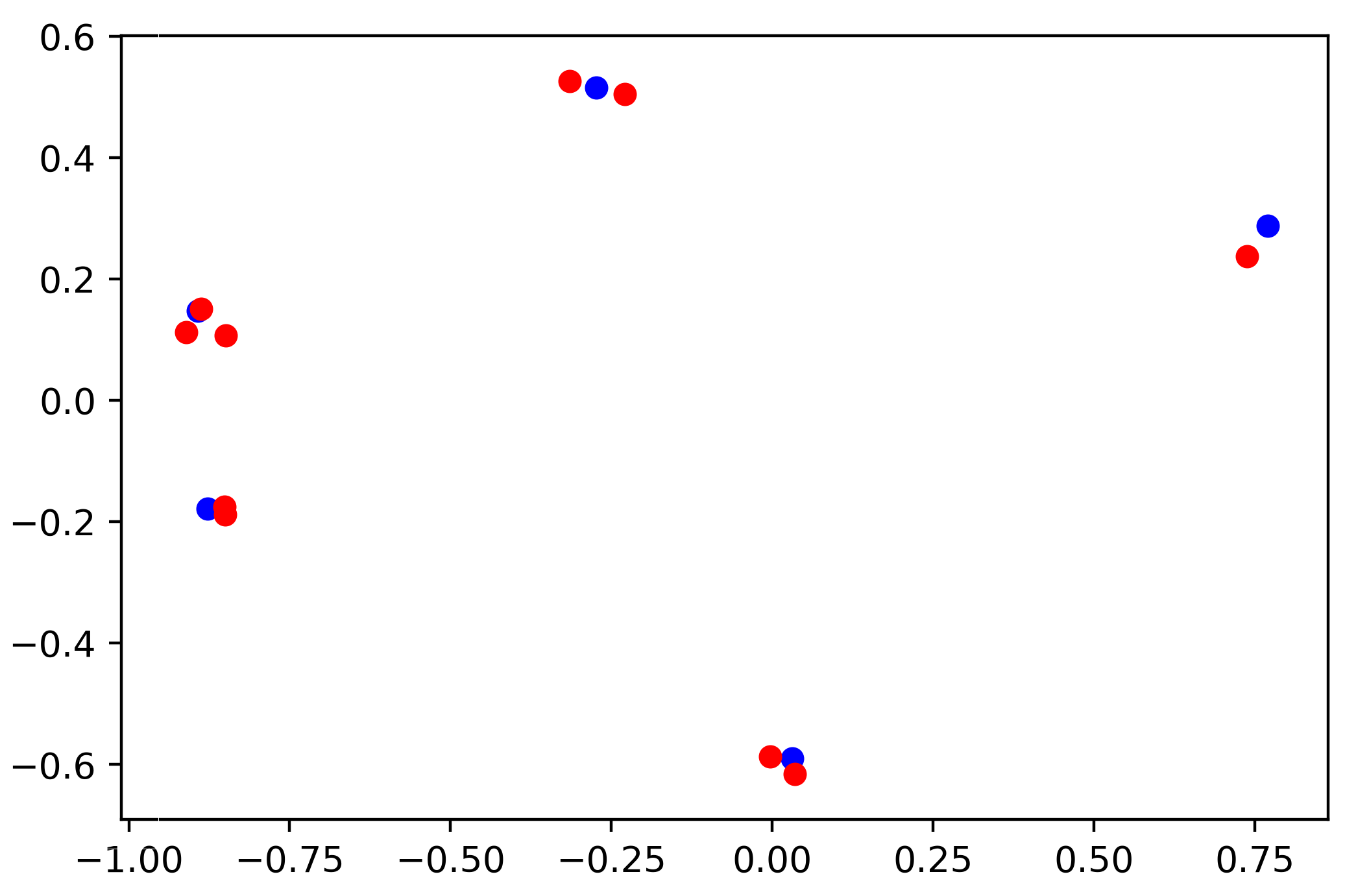}
\includegraphics[width=0.49\linewidth]{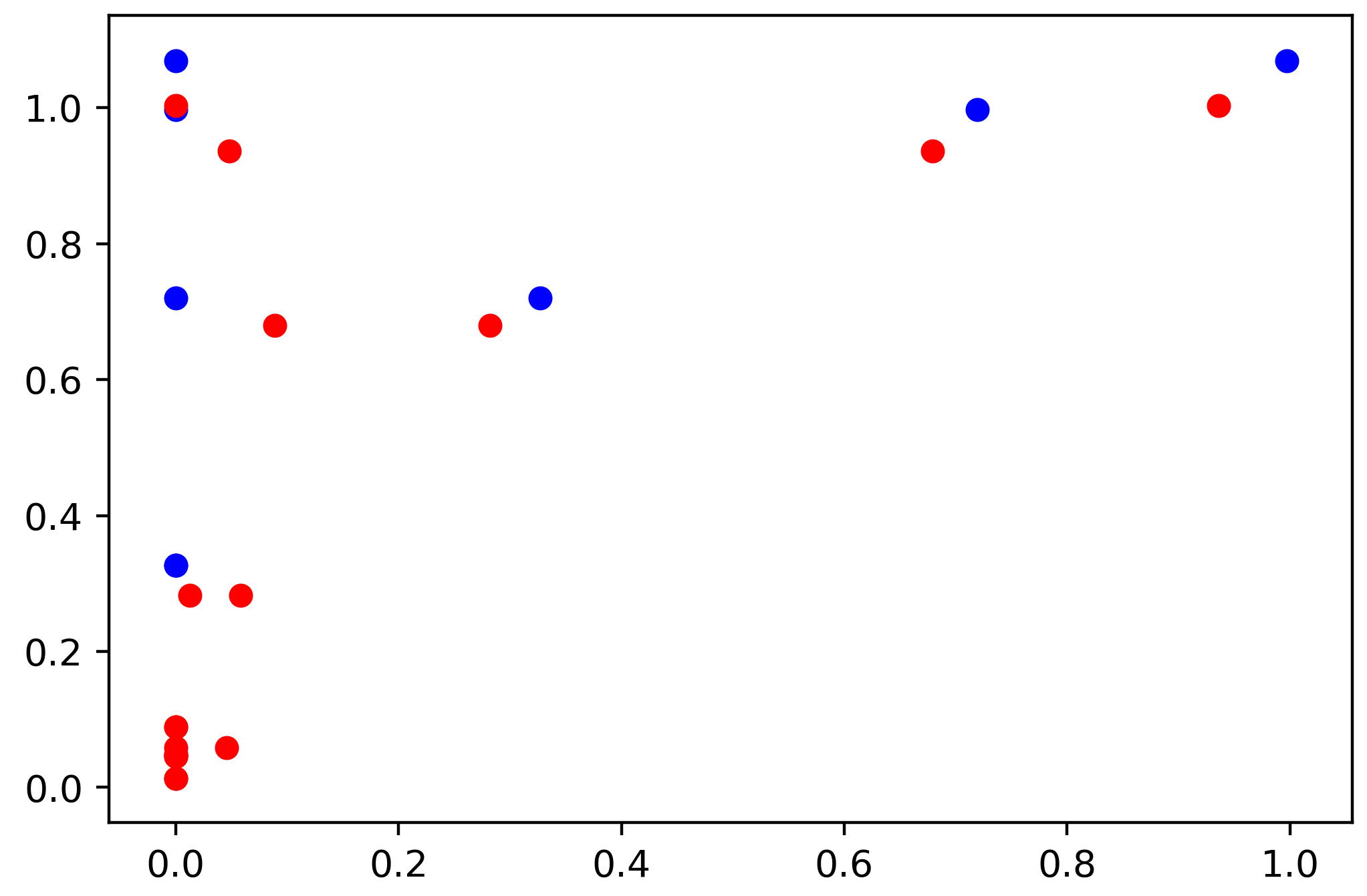}
\caption{\textbf{Left}: the cloud $C$ of 5 blue points is close to the cloud $C'$ of 10 red points in the Hausdorff distance. \textbf{Right}: the mergegrams are close in the bottleneck distance as predicted by Theorem~\ref{thm:stability_mergegram}.}
    \label{fig:perturbed_mergegram}
\end{figure}

Fig.~\ref{fig:perturbed_mergegram} illustrates Theorem~\ref{thm:stability_mergegram} on a cloud and its perturbation by showing their close mergegrams.
The more extensive experiment on 100 clouds  in \cite[Fig.~8]{elkin2020mergegram} similarly confirms that the mergegram is perturbed within expected bounds.  
The computational complexity of the mergegram $\MG(\De_{SL}(A))$ was proved to be near linear in the number $n$ of points in a cloud $A\subset\R^m$, see \cite[Theorem~8.2]{elkin2020mergegram}.
The results above justify that the invariant mergegram satisfies conditions (a,b,c) of Isometry Recognition Problem~\ref{pro:recognition}.

\section{New experiments on isometry recognition of substantially distorted real shapes}
\label{sec:experiments}

This section fulfills final condition (d) of Problem~\ref{pro:recognition} by experimentally comparing the mergegram with 0D persistence and distributions of distances to neighbors on 15000 clouds.
The earlier paper \cite{elkin2020mergegram} did experiments only on randomly generated clouds.
\medskip

We considered 15 classes of shapes represented by black and white images of mythical creatures \cite{bronstein2008analysis}, see Fig.~\ref{fig:myth_images}.
These shapes were chosen to make the shape recognition problem really challenging.
Indeed, similar creatures from this dataset are represented by slightly different shapes, which can be hard to isometrically distinguish from each other.
For example, several images of a horse include only minor differentiating features such as a saddle or a different tails, which makes horses nearly identical.
\medskip

\textbf{Shape generation}.
For each image, we generated 1000 perturbed images by affine and projective transformations to get 15000 distorted shapes split into 15 classes.
\medskip

First we rotated each image around its central point by an angle generated uniformly in the interval $[0,2\pi)$ using the function cv::rotate from the OpenCV library. 
If needed, we extended the resulting image to fit all black pixels of the rotated shape into a bounding box.
Then both affine and projective transformations distort each image by using a noise parameter $\de$ such that the value $\de=0$ represents the identity transformation.
\medskip

Fig.~\ref{fig:distorted_shapes_corner_points} illustrates how an original image is randomly rotated, then randomly distorted by affine or projective transformations depending on the noise parameter $\de$.
\medskip

\begin{figure}
\footnotesize
\centering
\input figures/distorted_shapes_corner_points.txt
\caption{Generating distorted shapes by applying random rotations, affine and projective transformations, which substantially affect the extracted clouds of Harris corner points \cite{sanchez2018analysis} in red.}
\label{fig:distorted_shapes_corner_points}
\end{figure}

\textbf{Affine transformations} are implemented as compositions of the already applied rotations above and the function cv::resize() from the OpenCV library.
This function scales an image of size $w\times h$ by horizontal and vertical factors $a,b$ sampled as follows.
\medskip

\noindent
$\bullet$
\textbf{Uniform noise}:
$a \in [1-\delta w, 1+\delta w]$, 
$b \in [1-\delta h, 1+\delta h]$ have uniform distributions.
\medskip

\noindent
$\bullet$
\textbf{Gaussian noise}: 
$a \in \mathcal{N}(1,\delta h) \cap \R_{+}$ and 
$b \in \mathcal{N}(1,\delta w) \cap \R_{+}$
have Gaussian distributions with mean 1 and standard variance $\de h, \de v$, truncated to positive numbers.
\medskip

\textbf{Projective transformation} are implemented as compositions of the already applied rotations above and the OpenCV function  cv::getPerspectiveTransform() function, which is parametrized by 4-dimensional array $v = (a_0, a_1, a_2, a_3)$ consisting of points $a_i\in\Z^2$, $i=0,1,2,3$.
This function maps the corners of the image as follows: $$(0,0)\mapsto a_0, (0,h) \mapsto a_1, (w, 0) \mapsto a_2  \text{ and } (w, h) \mapsto a_3.$$
Then the projective transformation of the rectangle $w\times h$ is uniquely determined by the above corners.
The above points $a_i$ are randomly sampled by using a noise parameter $\de$.
\medskip

\noindent
$\bullet$
\textbf{Uniform noise}: each coordinate has a uniform distribution with a noise parameter $\de$  
$$a_0 \in [0,\delta w] \times [0,\delta h],\qquad 
a_1 \in [0,\delta w] \times [h - \delta h, h],$$ 
$$a_2 \in [w - \delta w, w] \times [0,\delta h],\qquad 
a_3 \in [w - \delta w, w] \times [h - \delta h, h].$$

\noindent
$\bullet$
\textbf{Gaussian noise}: each coordinate has a Gaussian distribution truncated to the image
$$a_0 \in (\mathcal{N}(0,\delta w)\cap [0,w]) \times (\mathcal{N}(0,\delta h)\cap [0,h]),$$
$$a_1 \in (\mathcal{N}(0,\delta w)\cap [0,w]) \times (\mathcal{N}(h,\delta h)\cap [0,h]),$$ 
$$a_2 \in (\mathcal{N}(w,\delta w)\cap [0,w])\times(\mathcal{N}(0,\delta h)\cap [0,h]),$$
$$a_3 \in (\mathcal{N}(w,\delta w)\cap [0,w])\times(\mathcal{N}(h,\delta h)\cap [0,h]).$$

\textbf{Point cloud extraction}. For each distorted image, we extract classical Harris point corners \cite{sanchez2018analysis} due to their simplicity, see the red points in Fig.~\ref{fig:distorted_shapes_corner_points}.
For detecting corner points, the OpenCV function cv::cornerHarris was used with the parameters blockSize = 3, apertureSize = 5, k = 0.04, thresh = 120.
However one can use any reliable algorithms such as FAST \cite{rosten2008faster} or scale-invariant feature transform (SIFT) \cite{lowe1999object}. 
\medskip

After describing the available point cloud data above, we specify condition~(\ref{pro:recognition}d) of Isometry Recognition Problem~\ref{pro:recognition} in the context of supervised machine learning.

\begin{pro}[experimental recognition]
\label{pro:experimental}
Given a labeled dataset split into classes of similar but projectively distorted shapes, develop a supervised learning tool to recognize a class of distorted shapes with a high accuracy despite substantial noise.  
\end{pro}

Since all isometry invariants are independent of point ordering, the most suitable neural network is PersLay \cite{carriere2019perslay} whose output is invariant under permutations by design. 
Each layer is a combination of a coefficient layer $\omega(p):\mathbb{R}^m \rightarrow \mathbb{R}$, a transformation $\phi(p):\mathbb{R}^m \rightarrow \mathbb{R}^q$ and a permutation invariant layer $\text{op}$ combined as follows
$$\text{PersLay}(D) = \text{op}(\{w(p)\phi(p)\}_{p\in D}), \text{ where } D \text{ is a diagram or multiset of points in }\R^m.$$

Coordinates of all input points are linearly normalized to $[0,1]$.
We have used the following parameters of the PersLay network for all experiments below.
\medskip

\noindent    
\textbf{The max layer} $\text{MAX}(q)$ consists of the following functions.
    \begin{itemize}
        \item The coefficient layer $w : \mathbb{R}^m \rightarrow \mathbb{R}$ is the weight $w(x_1,\dots,x_m) = k|x_1-x_2|$, where $k$ is a trainable scalar and the dimension is typically $m=2$. 
        \item The transformation layer $\phi: \{\text{diagrams of points in }\mathbb{R}^m\} \rightarrow \mathbb{R}^q$ is the function $\phi(D) = \sum_{p \in D}\lambda p + \gamma\text{maxpool}(D) + \beta$, where $\lambda$,$\gamma$ are $\mathbb{R}^{m\times q}$ trainable matrices, $\beta$ is a $\mathbb{R}^q$ trainable vector and $\text{maxpool}$ returns a maximum value for every $i=1,\dots,m$.
\item The operational layer $\text{op}:\mathbb{R}^q \rightarrow \mathbb{R}^t$ puts all coordinates in increasing order and composes the result with standard densely connected layer \cite{tensorflow2015-whitepaper} $\text{Dense}: \mathbb{R}^q \rightarrow \mathbb{R}^t$.
    \end{itemize}

An output is a vector in $\R^t$ for $t=15$ of image classes.
A final prediction is obtained by choosing a class with a largest coordinate in the output vector.
\medskip

\noindent    
\textbf{The image layer} $\text{Im}[x,y]$ for integer parameters $x,y$ and a multiset of points in the unit square $[0,1]^2$ consists of the following functions.    
\begin{itemize}
        \item The coefficient layer $w:\R^2 \rightarrow \R$ is a piecewise constant function trained on $x\cdot y$ parameters, defined on the unit square partition 
        $$\mathcal{P}(x,y) = \left\{ \left[\frac{i}{x}, \frac{i+1}{x}\right] \times  \left[\frac{j}{y}, \frac{j+1}{y}\right]  \mid i = 0,\dots,x-1 \text{ and } j =0,\dots,y-1 \right\}. $$
\item Let $\phi_{p}:\mathbb{R}^2 \rightarrow \mathbb{R}$ be the Gaussian distribution centered at $p\in D$ with a trainable standard deviation $\sigma$. The transformation layer $\phi:\R \rightarrow \R^{xy}$ consists of $xy$ functions $\phi_p$, where $p$ runs over all centroids of the partition $\mathcal{P}(x,y)$.
 \item The operation layer op takes the sum over the given point cloud. A final prediction is made by composing the operation layer with the Dense layer.
    \end{itemize}
\medskip

Finally, the PersLay network used the optimizer tf.keras.adam with the standard learning rate 0.01 and 150 epochs, the loss function SparseCategoricalCrossEntropy, the 80:20 of training and testing,
 a 5-fold Monte Carlo cross validation for each run.
\medskip

Fig.~\ref{fig:max-layer_affine}, \ref{fig:max-layer_projective}, \ref{fig:image-layer_affine} show that the mergegram $\MG$ consistently outperforms two other isometry invariants: 0D persistence and the multiset $NN(4)$ consisting of 4-tuple distances to neighbors per given point. 
The simpler multiset $NN(2)$ performed worse.
A given cloud $C\subset\R^2$ was considered as a baseline input.
The noise factors $\de$ reached 25\%, which means that original images were distorted up to a quarter of image sizes. 

\section{A discussion of novel contributions and further open problems}
\label{sec:discussion}

This paper has further demonstrated that the provably stable-under-noise invariant mergegram of a dendrogram is a fast and efficient tool in the challenging problem of isometry shape recognition, especially for substantially distorted images.
\medskip

In comparison with the conference version \cite{elkin2020mergegram}, section~\ref{sec:relations} proved new Theorem~\ref{thm:mergegram-to-dendrogram} describing how to reconstruct a single-linkage dendrogram in general position from its much simpler mergegram.   
It is hard to define a continuous metric between dendrograms, especially because they can be unstable under perturbations.
Theorem~\ref{thm:mergegram-to-dendrogram} allows us to measure a continuous similarity between dendrograms in general position as the bottleneck distance between their unique mergegrams.
This distance can be computed in time $O(n^{1.5}\log n)$ \cite{kerber2016geometry} for diagrams consisting of at most $n$ points.
\medskip

Section~\ref{sec:stability} provided a full proof of stability of the mergegram under perturbations of points, while the earlier paper \cite{elkin2020mergegram} only announced this result without proving highly non-trivial Lemmas~\ref{lem:merge_module_decomposition} and~\ref{lem:merge_modules_interleaved}, which required a heavy algebraic machinery.
\medskip

\begin{figure}[h!]
\centering
\input figures/max-layer_affine_uniform.txt
\input figures/max-layer_affine_Gaussian.txt
\caption{Recognition rates are obtained by training the max layer MAX(75) of PersLay on three isometry invariants and a cloud of corner points extracted from 15000 affinely distorted images.}
\label{fig:max-layer_affine}
\end{figure}

\begin{figure}[h!]
\centering
\input figures/max-layer_projective_uniform.txt
\input figures/max-layer_projective_Gaussian.txt
\caption{Recognition rates are obtained by training the max layer MAX(75) of PersLay on isometry invariants and corner points extracted from 15000 projectively distorted images.}
\label{fig:max-layer_projective}
\end{figure}

\begin{figure}[h!]
\centering
\input figures/image-layer_affine_uniform.txt
\input figures/image-layer_affine_Gaussian.txt
\caption{Recognition rates are obtained by training the image layer IM[20,20] of PersLay on isometry invariants and a cloud of corner points extracted from 15000 affinely distorted images.}
\label{fig:image-layer_affine}
\end{figure}

Example~\ref{exa:5-point_line} and the discussion following Theorem~\ref{thm:mergegram_to_0D_persistence} justify that the invariant mergegram is strictly stronger than the 0D persistence.
This theoretical fact is now confirmed by the new experiments on 15000 point clouds extracted from substantially distorted real shapes.
In Fig.~\ref{fig:max-layer_affine}, \ref{fig:max-layer_projective}, \ref{fig:image-layer_affine} the mergegram outperformed other isometry invariants.
Since the distribution $NN(2)$ of distances to two closest neighbors per point performed badly, we have strengthened this invariant to $NN(4)$ of distances to four nearest neighbors.
However, even $NN(4)$ performed always always worse than the original point cloud, which can not be considered as an isometry invariant.
\medskip

For very high levels of 20\% and 25\% distortions in projective transformations, the PersLay network trained on a point cloud achieved high recognition rates, because we have extensively tried many parameters in the layers MAX(75) and Im[20,20] for a best trade-off between accuracy and speed.
The C++ code for the mergegram is at \cite{elkin2020mergegram}.
\medskip

We thank all reviewers in advance for their valuable time and helpful suggestions.   

\bibliographystyle{plainurl}
\bibliography{recognition-mergegram}

\end{document}

%% file: figures/5-point_line.txt
\begin{figure}[h!]
\centering
\begin{tikzpicture}[scale = 1.1]
  \draw[->] (-1,0) -- (11,0) node[right]{} ;
   \foreach \x/\xtext in {0, 1, 3, 7, 10}
    \draw[shift={(\x,0)}] (0pt,2pt) -- (0pt,-2pt) node[below] {$\xtext$};   
   \filldraw (0,0) circle (2pt);
   \filldraw (1,0) circle (2pt);
   \filldraw (3,0) circle (2pt);
   \filldraw (7,0) circle (2pt);
   \filldraw (10,0) circle (2pt);
\end{tikzpicture}

\begin{tikzpicture}[thick,scale=0.50, every node/.style={transform shape}][sloped]
\draw[style=help lines,step = 1] (-1,0) grid (10.4,4.4);
\draw [->] (-1,0) -- (-1,5) node[above] {{\large scale $s$}};
\foreach \i in {0,0.5,...,2}{ \node at (-1.5,2*\i) {\i}; }
\node (a) at (0,-0.3) {0};
\node (b) at (1,-0.3) {1};
\node (c) at (3,-0.3) {3};
\node (d) at (7,-0.3) {7};
\node (e) at (10,-0.3) {10};
\node (x) at (5,5) {};
\node (ab) at (0.5,1){};
\node (abc) at (1.5,2){};
\node (de) at (8.5,3){};
\node (all) at (5,4){};
\draw [line width=0.5mm, blue ] (a) |- (ab.center);
\draw [line width=0.5mm, blue ] (b) |- (ab.center);
\draw [line width=0.5mm, blue ] (c) |- (abc.center);
\draw [line width=0.5mm, blue ] (d) |- (de.center);
\draw [line width=0.5mm, blue ] (e) |- (de.center);
\draw [line width=0.5mm, blue ] (ab.center) |- (abc.center);
\draw [line width=0.5mm, blue ] (abc.center) |- (all.center);
\draw [line width=0.5mm, blue ] (de.center) |- (all.center);
\draw [line width=0.5mm, blue ] [->] (all.center) -> (x.center);
\end{tikzpicture}
\hspace*{1mm}
\begin{tikzpicture}[thick,scale=0.9, every node/.style={transform shape}]
  \draw[style=help lines,step = 0.5] (0,0) grid (0.5,2.4);
  \draw[->] (-0.2,0) -- (0.8,0) node[right] {birth};
  \draw[->] (0,-0.2) -- (0,2.4) node[above] {};	
  \draw[-] (0,0) -- (1,1) node[right]{};
  \foreach \x/\xtext in {0.5/0.5}
    \draw[shift={(\x,0)}] (0pt,2pt) -- (0pt,-2pt) node[below] {$\xtext$};
  \foreach \y/\ytext in {0.5/0.5,  1/1, 1.5/1.5, 2.0/2}
    \draw[shift={(0,\y)}] (2pt,0pt) -- (-2pt,0pt) node[left] {$\ytext$};  
   \draw [blue,fill] (0,0.5) circle (2pt);
   \draw [blue,fill] (0.0,1) circle (2pt);
   \draw [blue,fill] (0,1.5) circle (2pt);
   \draw [blue,fill] (0,2) circle (2pt);
   \draw [blue,fill] (0,2.7) circle (2pt);
\end{tikzpicture}
\hspace*{1mm}
\begin{tikzpicture}[thick,scale=0.9, every node/.style={transform shape}]
  \draw[style=help lines,step = 0.5] (0,0) grid (2.4,2.4);
  \draw[->] (-0.2,0) -- (2.4,0) node[right] {birth};
  \draw[->] (0,-0.2) -- (0,2.4) node[above] {death};	
  \draw[-] (0,0) -- (2.4,2.4) node[right]{};
  \foreach \x/\xtext in {0.5/0.5, 1/1, 1.5/1.5, 2.0/2}
    \draw[shift={(\x,0)}] (0pt,2pt) -- (0pt,-2pt) node[below] {$\xtext$};
  \foreach \y/\ytext in {0.5/0.5,  1/1, 1.5/1.5, 2.0/2}
    \draw[shift={(0,\y)}] (2pt,0pt) -- (-2pt,0pt) node[left] {$\ytext$};  
   \draw [red,fill] (0,0.5) circle (2pt);
   \draw [red] (0,0.5) circle (4pt);
   \draw [red,fill] (0.0,1.0) circle (2pt);
   \draw [red,fill] (0.0,1.5) circle (2pt);
   \draw [red] (0,1.5) circle (4pt);
   \draw [red,fill] (0.5,1.0) circle (2pt);
   \draw [red,fill] (1,2) circle (2pt);
   \draw [red,fill] (1.5, 2.0) circle (2pt);
   \draw [red,fill] (2, 2.7) circle (2pt);
\end{tikzpicture}
\caption{\textbf{Top}: the 5-point cloud $A = \{0,1,3,7,10\}\subset\R$.
\textbf{Bottom} from left to right: single-linkage dendrogram $\De_{SL}(A)$ from Definition~\ref{dfn:sl_clustering}, the 0D persistence diagram $\PD$ from Definition~\ref{dfn:persistence_diagram} and the new mergegram $\MG$ from Definition~\ref{dfn:mergegram}, where double circles show pairs of multiplicity 2.}
\label{fig:5-point_line}
\end{figure}

%% file: figures/3-point_dendrogram.txt
\parbox{90mm}{
\footnotesize
\begin{tabular}{lccccc}
partition $\De(A;2)$ at scale $s_2=2$ & & & $\{0,1,2\}$ & & \\
map $\De_1^2:\De(A;1)\to\De(A;2)$ & & & $\uparrow$ & $\nwarrow$ & \\
partition $\De(A;1)$ at scale $s_1=1$ & & & \{0, 1\} & & \{2\} \\
map $\De_0^1:\De(A;0)\to\De(A;1)$ & & $\nearrow$ & $\uparrow$ & & $\uparrow$  \\
partition $\De(A;0)$ at scale $s_0=0$ & $\{0\}$ & & $\{1\}$ & & \{2\} 
\end{tabular}}
\parbox{25mm}{
\begin{tikzpicture}[thick,scale=0.75, every node/.style={transform shape}]
  \draw[style=help lines,step = 1] (0,0) grid (2.4,2.4);
  \draw[->] (-0.2,0) -- (2.4,0) node[right] {birth};
  \draw[->] (0,-0.2) -- (0,2.4) node[above] {death};	
  \draw[-] (0,0) -- (2.4,2.4) node[right]{};
  \foreach \x/\xtext in {1/1, 2/2}
    \draw[shift={(\x,0)}] (0pt,2pt) -- (0pt,-2pt) node[below] {$\xtext$};
  \foreach \y/\ytext in {1/1, 2/2}
    \draw[shift={(0,\y)}] (2pt,0pt) -- (-2pt,0pt) node[left] {$\ytext$};
   \draw [red,fill] (0,1) circle (2pt);
   \draw [red] (0,1) circle (4pt);
   \draw [red,fill] (0,2) circle (2pt);
   \draw [red,fill] (1,2) circle (2pt);
   \draw [red,fill] (2,2.7) circle (2pt);
\end{tikzpicture}}

%% file: figures/5-point_set.txt
\parbox{80mm}{
\begin{tikzpicture}[scale = 0.75][sloped]
\node (x) at (5,3) {x};
 \node (a) at (1,1) {a};
 \draw (a) -- node[above]{3} ++ (x);
 \node (b) at (3.5,4.0) {b};
 \draw (b) -- node[above]{1} ++ (x);
 \node (c) at (7,1) {c};
 \draw (c) -- node[below]{1} ++ (x);
 \node (y) at (8,3) {y};
 \draw (x) -- node[above]{2} ++ (y);
 \node (d) at (10,5){p};
 \node (e) at (10,1){q};
 \draw (y) -- node[below]{1} ++ (d);
 \draw (y) -- node[below]{1} ++ (e);
\end{tikzpicture}}
\parbox{40mm}{
\begin{tabular}{c|ccccc}
& a & b & c & p & q \\
\hline
a & 0 & 4 & 4 & 6 & 6 \\
b & 4 & 0 & 4 & 6 & 6 \\
c & 4 & 2 & 0 & 4 & 4 \\
p & 6 & 4 & 4 & 0 & 2 \\
q & 6 & 4 & 4 & 2 & 0 
\end{tabular}}

%% file: figures/5-point_mergegram.txt
\begin{tikzpicture}[scale = 0.6][sloped]
 \draw[style=help lines,step = 1] (-1,0) grid (10.4,6.3);
\foreach \i in {0,...,3.0} { \node at (-1.4,2*\i) {\i}; }
\node (a) at (0,-0.3) {a};
\node (b) at (4,-0.3) {b};
\node (c) at (6,-0.3) {c};
\node (d) at (8,-0.3) {p};
\node (e) at (10,-0.3) {q};
\node (x) at (5,6.75) {};

\node (de) at (9,2){};
\node (bc) at (5.0,2){};
\node (bcde) at (7.0,4){};
\node (all) at (5.0,6){};

\draw [line width=0.5mm, blue ] (a) |- (all.center);
\draw [line width=0.5mm, blue ] (b) |- (bc.center);
\draw [line width=0.5mm, blue ]  (c) |- (bc.center);
\draw [line width=0.5mm, blue ]  (d) |- (de.center);
\draw [line width=0.5mm, blue ]  (e) |- (de.center);
\draw [line width=0.5mm, blue ]  (de.center) |- (bcde.center);
\draw [line width=0.5mm, blue ]  (bc.center) |- (bcde.center);
\draw [line width=0.5mm, blue ]  (bcde.center) |- (all.center);
\draw [line width=0.5mm, blue] [->] (all.center) -> (x.center);
\end{tikzpicture}
\hspace*{1cm}
\begin{tikzpicture}[scale = 1.1]
 \draw[style=help lines,step = 1] (0,0) grid (3.4,3.4);
 \draw[->] (-0.2,0) -- (3.4,0) node[right] {birth};
  \draw[->] (0,-0.2) -- (0,3.4) node[above] {death};	
  \draw[-] (0,0) -- (3.4,3.4) node[right]{};

  \foreach \x/\xtext in {1/1, 2.0/2, 3.0/3}
    \draw[shift={(\x,0)}] (0pt,2pt) -- (0pt,-2pt) node[below] {$\xtext$};

  \foreach \y/\ytext in {1/1, 2.0/2, 3.0/3}
    \draw[shift={(0,\y)}] (2pt,0pt) -- (-2pt,0pt) node[left] {$\ytext$}; 
   \draw [red,fill] (0,1) circle (2pt);
   \draw [red] (0,1) circle (3pt);
   \draw [red] (0,1) circle (4pt);
   \draw [red] (0,1) circle (5pt);
   \draw [red,fill] (1,2) circle (2pt);
   \draw [red] (1,2) circle (4pt);
   \draw [red,fill] (2,3) circle (2pt);
\draw [red,fill] (0,3) circle (2pt);
\draw [red,fill] (3,3.7) circle (2pt);
\end{tikzpicture}

%% file: figures/distorted_shapes_corner_points.txt
    \begin{tikzpicture}[ >=stealth,
    node distance=4cm,
    on grid,
    auto
  ]
        \node [] (start) {\includegraphics[scale = 0.25]{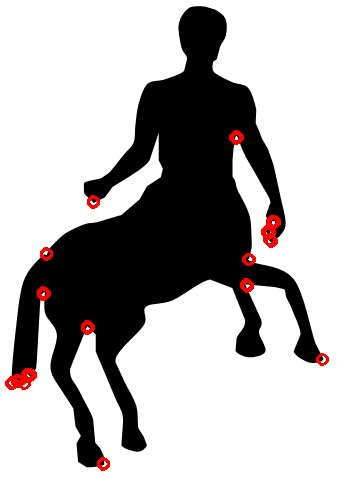} };
        \node [below = 2.0cm of start] (startText) {(1) Original centaur};
      \node [right of= start] (A1) {\includegraphics[scale = 0.25]{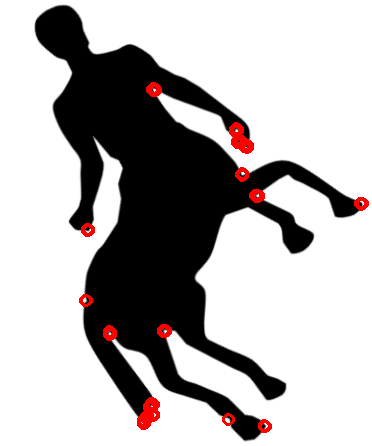}};
        \node [below = 2.0cm of A1] (A1Text) {(2) Rotated centaur };
    	\node [right of= A1] (tmp) {};
        \node [above = 2.0cm of tmp ] (A3){\includegraphics[scale = 0.25]{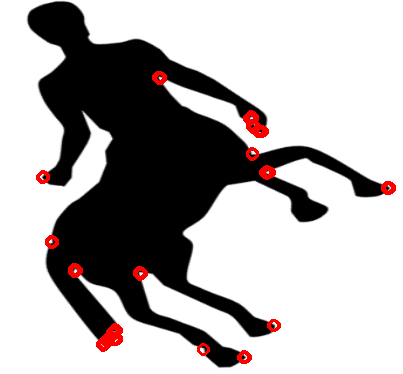}};
                \node [below = 2.0cm of A3] (A3Text) {(3a) Affine transformation };
        \node [below = 2.0cm of tmp] (end){\includegraphics[scale = 0.25]{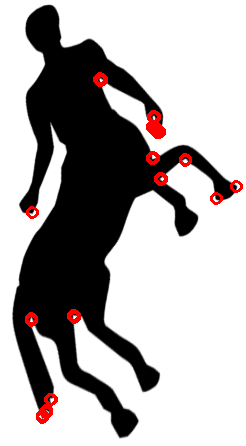}};
                \node [below = 2.0cm of end] (endText) {(3b) Projective transformation };
       \path[draw, ->] (start) -- (A1);
        \path[draw, ->] (A1) -- (A3);
        \path[draw, ->] (A1) -- (end);                    
    \end{tikzpicture}
    \begin{tikzpicture}[ >=stealth,
    node distance=4cm,
    on grid,
    auto
  ]
        \node [] (start) {\includegraphics[scale = 0.25]{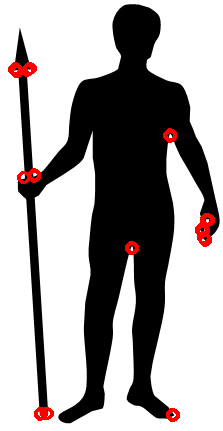} };
        \node [below = 2.0cm of start] (startText) {(1) Original man};
        \node [right of= start] (A1) {\includegraphics[scale = 0.25]{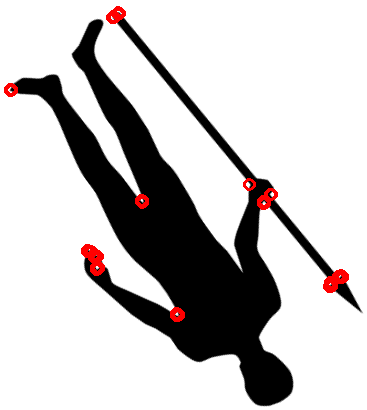}};
        \node [below = 2.0cm of A1] (A1Text) {(2) Rotated man };
        \node [right of= A1] (tmp) {};
        
        \node [above = 2.0cm of tmp ] (A3){\includegraphics[scale = 0.25]{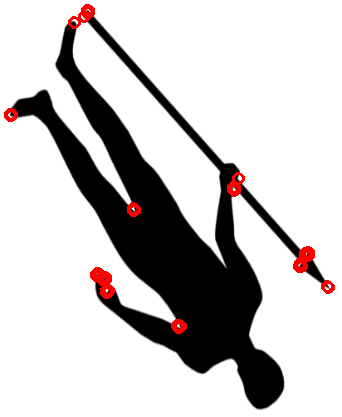}};
                \node [below = 2.0cm of A3] (A3Text) {(3a) Affine transformation };
        \node [below = 2.0cm of tmp] (end){\includegraphics[scale = 0.25]{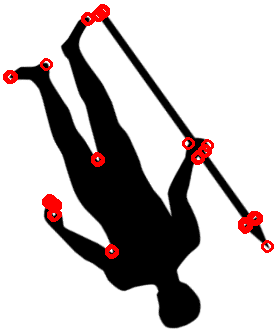}};
                \node [below = 2.0cm of end] (endText) {(3b) Projective transformation };
       \path[draw, ->] (start) -- (A1);
        \path[draw, ->] (A1) -- (A3);
        \path[draw, ->] (A1) -- (end);                    
    \end{tikzpicture}
    \begin{tikzpicture}[ >=stealth,
    node distance=4cm,
    on grid,
    auto
  ]
        \node [] (start) {\includegraphics[scale = 0.25]{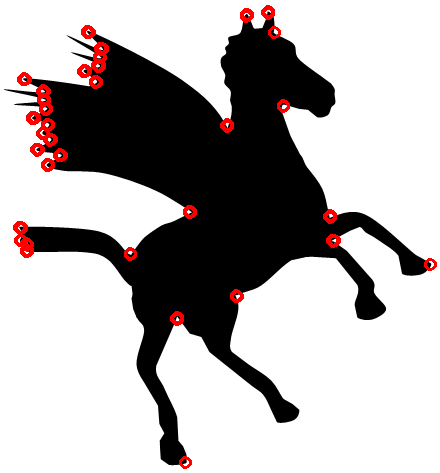} };
        \node [below = 2.0cm of start] (startText) {(1) Original horse};
        \node [right of= start] (A1) {\includegraphics[scale = 0.25]{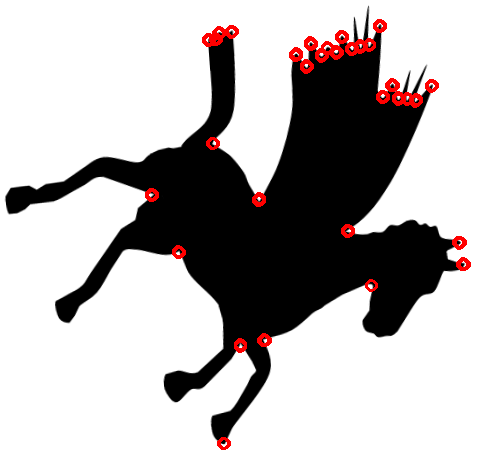}};
        \node [below = 2.0cm of A1] (A1Text) {(2) Rotated horse };
    	\node [right of= A1] (tmp) {};        
        \node [above = 2.0cm of tmp ] (A3){\includegraphics[scale = 0.25]{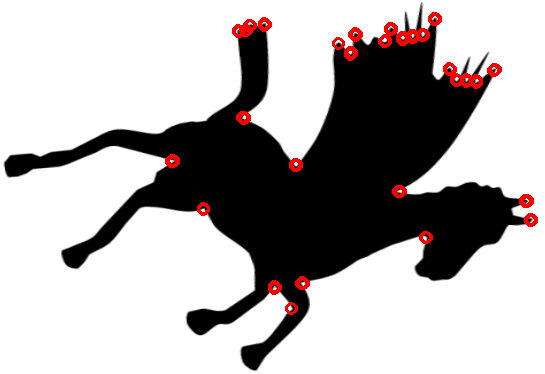}};
                \node [below = 2.0cm of A3] (A3Text) {(3a) Affine transformation };
        \node [below = 2.0cm of tmp] (end){\includegraphics[scale = 0.25]{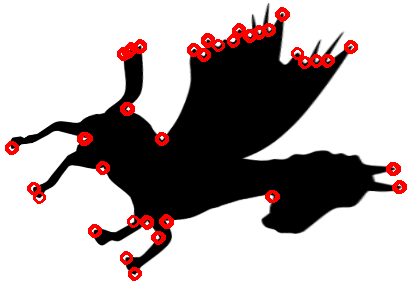}};
                \node [below = 2.0cm of end] (endText) {(3b) Projective transformation };         
       \path[draw, ->] (start) -- (A1);
        \path[draw, ->] (A1) -- (A3);
        \path[draw, ->] (A1) -- (end);                    
    \end{tikzpicture}

%% file: figures/max-layer_affine_uniform.txt
\begin{tikzpicture}[thick,scale=0.75, every node/.style={transform shape}]
\begin{axis}[xlabel = bound $\de$ of uniform noise relative to image size, ylabel = recognition rate, grid, grid style={gray},legend style={at={(0.5,-0.1)},anchor=north}]
\addplot[teal, mark=*] table [x=e, y=m, col sep=comma] {Tables/ProjectiveUniform2.csv};
\addplot table [x=e, y=p, col sep=comma] {Tables/ProjectiveUniform2.csv};
\addplot table [x=e, y=nn4, col sep=comma] {Tables/ProjectiveUniform.csv};
\addplot[blue,mark=*] table [x=e, y=c, col sep=comma] {Tables/ProjectiveUniform2.csv};
\end{axis}
\end{tikzpicture}

%% file: figures/max-layer_affine_Gaussian.txt
\begin{tikzpicture}[thick,scale=0.75, every node/.style={transform shape}]
\begin{axis}[xlabel = standard deviation $\de$ of Gaussian noise, grid,grid style={gray},legend style={at={(0.5,0.5)},anchor=center}]
\addplot[teal, mark=*] table [x=e, y=m, col sep=comma] {Tables/ScaleGaussian.csv};
\addlegendentry{mergegram}
\addplot table [x=e, y=p, col sep=comma] {Tables/ScaleGaussian.csv};
\addlegendentry{0D persistence}
\addplot table [x=e, y=nn4, col sep=comma] {Tables/ScaleGaussian.csv};
\addlegendentry{NN(4) distances}
\addplot[blue,mark=*] table [x=e, y=c, col sep=comma] {Tables/ScaleGaussian.csv};
\addlegendentry{point cloud}
\end{axis}
\end{tikzpicture}

%% file: figures/max-layer_projective_uniform.txt
\begin{tikzpicture}[thick,scale=0.75, every node/.style={transform shape}]
\begin{axis}[xlabel = upper bound $\de$ of uniform noise, ylabel = recognition rate,grid,grid style={gray},legend style={at={(0.5,-0.1)},anchor=north}]
\addplot[teal, mark=*] table [x=e, y=m, col sep=comma] {Tables/ProjectiveUniform2.csv};
\addplot table [x=e, y=p, col sep=comma] {Tables/ProjectiveUniform2.csv};
\addplot table [x=e, y=nn4, col sep=comma] {Tables/ProjectiveUniform.csv};
\addplot[blue,mark=*] table [x=e, y=c, col sep=comma] {Tables/ProjectiveUniform2.csv};
\end{axis}
\end{tikzpicture}

%% file: figures/max-layer_projective_Gaussian.txt
\begin{tikzpicture}[thick,scale=0.75, every node/.style={transform shape}]
\begin{axis}[xlabel = standard deviation $\de$ of Gaussian noise, grid, grid style={gray},legend style={at={(0.25,0.32)},anchor=north}]
\addplot[teal, mark=*] table [x=e, y=m, col sep=comma] {Tables/ProjectiveGaussian.csv};
\addlegendentry{mergegram}
\addplot table [x=e, y=p, col sep=comma] {Tables/ProjectiveGaussian.csv};
\addlegendentry{0D persistence}
\addplot table [x=e, y=nn4, col sep=comma] {Tables/ProjectiveGaussian.csv};
\addlegendentry{NN(4) distances}
\addplot[blue,mark=*] table [x=e, y=c, col sep=comma] {Tables/ProjectiveGaussian.csv};
\addlegendentry{point cloud}
\end{axis}
\end{tikzpicture}

%% file: figures/image-layer_affine_uniform.txt
\begin{tikzpicture}[thick,scale=0.75, every node/.style={transform shape}]
\begin{axis}[xlabel = upper bound $\de$ of uniform noise, ylabel = recognition rate, grid, grid style={gray},legend style={at={(0.5,-0.1)},anchor=north}]
\addplot[teal, mark=*] table [x=e, y=m, col sep=comma] {Tables/RotateScaleImage.csv};
\addplot table [x=e, y=p, col sep=comma] {Tables/RotateScaleImage.csv};
\addplot table [x=e, y=n, col sep=comma] {Tables/RotateScaleImage.csv};
\addplot[blue,mark=*] table [x=e, y=c, col sep=comma] {Tables/RotateScaleImage.csv};
\end{axis}
\end{tikzpicture}

%% file: figures/image-layer_affine_Gaussian.txt
\begin{tikzpicture}[thick,scale=0.75, every node/.style={transform shape}]
\begin{axis}[xlabel = standard deviation $\de$ of Gaussian noise, grid, grid style={gray} ,legend style={at={(0.75,0.45)},anchor=center}]
\addplot[teal, mark=*] table [x=e, y=m, col sep=comma] {Tables/ScaleGaussianImage.csv};
\addlegendentry{mergegram}
\addplot table [x=e, y=p, col sep=comma] {Tables/ScaleGaussianImage.csv};
\addlegendentry{0D persistence}
\addplot table [x=e, y=n, col sep=comma] {Tables/ScaleGaussianImage.csv};
\addlegendentry{NN(4) distances}
\addplot[blue,mark=*] table [x=e, y=c, col sep=comma] {Tables/ScaleGaussianImage.csv};
\addlegendentry{point cloud}
\end{axis}
\end{tikzpicture}